\documentclass[10pt,twocolumn,journal]{IEEEtran}
\usepackage{mathtools}
\usepackage{amssymb}
\usepackage{color}
\usepackage{multirow}
\usepackage{stmaryrd}
\usepackage{yfonts}
\usepackage{mathabx}
\usepackage{caption}
\usepackage{subcaption}
\usepackage{float}
\usepackage{stfloats}
\usepackage{graphicx}
\usepackage{tcolorbox}

\usepackage{tabularx}
\usepackage{booktabs}
\usepackage{caption}

\usepackage[numbers,sort&compress]{natbib}
\usepackage{multicol}
\usepackage{algorithm}
\usepackage{algpseudocode}
\usepackage{pifont}
\usepackage{varwidth}

\usepackage{afterpage}
\usepackage{balance}
\usepackage{lipsum}
\usepackage{mathbbol}

\def\mindex#1{\index{#1}}

\def\sq{\hbox{\rlap{$\sqcap$}$\sqcup$}}
\def\qed{\ifmmode\sq\else{\unskip\nobreak\hfil
\penalty50\hskip1em\null\nobreak\hfil\sq
\parfillskip=0pt\finalhyphendemerits=0\endgraf}\fi\medskip}

\long\def\defbox#1{\framebox[.9\hsize][c]{\parbox{.85\hsize}{%
\parindent=0pt
\baselineskip=12pt plus .1pt      % STYLE
\parskip=6pt plus 1.5pt minus 1pt % CHANGES
 #1}}}

\long\def\beginbox#1\endbox{\subsection*{}%
\hbox{\hspace{.05\hsize}\defbox{\medskip#1\bigskip}}%
\subsection*{}}

\def\endbox{}

\def\diag{{\text{diag}}}

\def\tr{\mathsf{tr}}

\newsavebox{\junk}
\savebox{\junk}[1.6mm]{\hbox{$|\!|\!|$}}

\def\bC{{\mathbb C}}

\def\bE{{\mathbb E}}

\def\bR{{\mathbb R}}

\def\sfH{{\sf H}}

\def\bfmath#1{{\mathchoice{\mbox{\boldmath$#1$}}%
{\mbox{\boldmath$#1$}}%
{\mbox{\boldmath$\scriptstyle#1$}}%
{\mbox{\boldmath$\scriptscriptstyle#1$}}}}

\def\bfmY{\bfmath{Y}}

\def\bfmhhaY{\bfmath{\hhaY}} %\widehat{\widehat{Y}}}}
\def\bfmhhaY{\hbox to 0pt{$\widehat{\bfmY}$\hss}\widehat{\phantom{\raise 1.25pt\hbox{$\bfmY$}}}}

\def\til={{\widetilde =}}

 \def\FRAC#1#2#3{\genfrac{}{}{}{#1}{#2}{#3}}

\def\ddtp{{\mathchoice{\FRAC{1}{d^{\hbox to 2pt{\rm\tiny +\hss}}}{dt}}%
{\FRAC{1}{d^{\hbox to 2pt{\rm\tiny +\hss}}}{dt}}%
{\FRAC{3}{d^{\hbox to 2pt{\rm\tiny +\hss}}}{dt}}%
{\FRAC{3}{d^{\hbox to 2pt{\rm\tiny +\hss}}}{dt}}}}

\def\average#1,#2,{{1\over #2} \sum_{#1}^{#2}}

\def\eye(#1){{\bf(#1)}\quad}

\newtheorem{theorem}{{\bf Theorem}}

\newtheorem{remark}{{\bf Remark}}

\newtheorem{lemma}{{\bf Lemma}}

\def\eq#1/{(\ref{e:#1})}

\newcommand{\beqn}[1]{\notes{#1}%
\begin{eqnarray} \elabel{#1}}

\newcommand{\eeqn}{\end{eqnarray} }

\newcommand{\beq}[1]{\notes{#1}%
\begin{equation}\elabel{#1}}

\newcommand{\eeq}{\end{equation}}

\def\bdes{\begin{description}}
\def\edes{\end{description}}

\newcounter{rmnum}

\newcounter{anum}

{\end{list}}

\def\ass(#1:#2){(#1\ref{#1:#2})}

\def\ritem#1{
\item[{\sf \ass(\current_model:#1)}]
}

\newenvironment{recall-ass}[1]{%
\begin{description}
\def\current_model{#1}}{
\end{description}
}

\long\def\comment#1{}

\newfont{\bb}{msbm10 scaled 1100}
\newcommand{\CC}{\mbox{\bb C}}

\newcommand{\av}{{\bf a}}
\newcommand{\bv}{{\bf b}}

\newcommand{\dv}{{\bf d}}
\newcommand{\ev}{{\bf e}}

\newcommand{\hv}{{\bf h}}

\newcommand{\nv}{{\bf n}}

\newcommand{\qv}{{\bf q}}

\newcommand{\sv}{{\bf s}}

\newcommand{\uv}{{\bf u}}
\newcommand{\wv}{{\bf w}}
\newcommand{\vv}{{\bf v}}
\newcommand{\xv}{{\bf x}}
\newcommand{\yv}{{\bf y}}
\newcommand{\zv}{{\bf z}}

\newcommand{\Cm}{{\bf C}}

\newcommand{\Fm}{{\bf F}}
\newcommand{\Gm}{{\bf G}}
\newcommand{\Hm}{{\bf H}}
\newcommand{\Id}{{\bf I}}

\newcommand{\Qm}{{\bf Q}}

\newcommand{\Sm}{{\bf S}}

\newcommand{\Um}{{\bf U}}
\newcommand{\Wm}{{\bf W}}
\newcommand{\Vm}{{\bf V}}
\newcommand{\Xm}{{\bf X}}

\newcommand{\Ac}{{\cal A}}

\newcommand{\Cc}{{\cal C}}

\newcommand{\Kc}{{\cal K}}

\newcommand{\Nc}{{\cal N}}

\newcommand{\alphav}{\hbox{\boldmath$\alpha$}}

\newcommand{\Lambdam}{\hbox{\boldmath$\Lambda$}}
\newcommand{\Deltam}{\boldsymbol{\Delta}}
\newcommand{\Sigmam}{\hbox{\boldmath$\Sigma$}}

\newcommand{\Psim}{\hbox{\boldmath$\Psi$}}

\newcommand{\trace}{{\hbox{tr}}}

\newcommand{\transp}{{\sf T}}

\def\herm{{\sfH}}

\newcommand{\snrul}{{\sf snr}_{\rm ul}}
\newcommand{\snrdl}{{\sf snr}_{\rm dl}}

\newcommand{\betatr}{\beta_{\rm tr}}
\newcommand{\betafb}{\beta_{\rm fb}}

\usepackage{tikz, pgfplots,graphicx,xcolor}
\usetikzlibrary{plotmarks,spy,backgrounds}
\pgfplotsset{compat=newest}

\usepackage{makecell}

\allowdisplaybreaks

\newcommand{\normd}[1]{{\left\vert\kern-0.25ex\left\vert\kern-0.25ex\left\vert #1 
		\right\vert\kern-0.25ex\right\vert\kern-0.25ex\right\vert}}

\title{Downlink CSIT under Compressed Feedback: Joint vs. Separate Source-Channel Coding}
 \author{
	\IEEEauthorblockN{Yi Song,
        Tianyu Yang, \emph{Member, IEEE}, Mahdi Barzegar Khalilsarai,
	and Giuseppe Caire, \emph{Fellow, IEEE}
	}
	\thanks{
The work of Y. Song and G. Caire was supported by BMBF Germany in the program of ``Souverän. Digital. Vernetzt.'' Joint Project 6G-RIC (Project IDs 16KISK030).

Y. Song, T. Yang and G. Caire are with the Faculty of Electrical Engineering and Computer Science at the Technical University of Berlin, 10587 Berlin, Germany (e-mail: yi.song@tu-berlin.de; tianyu.yang@tu-berlin.de; caire@tu-berlin.de).
    
M. Khalilsarai is with the department of Wireless Research and Standards at Lenovo Research, Stuttgart, Germany (e-mail: mkhalilsarai@lenovo.com).
	}
}

\begin{document}
\maketitle

% \vspace{-1cm}

\begin{abstract}
     The acquisition of Downlink (DL) channel state information at the transmitter (CSIT) is known to be a challenging task in multiuser massive MIMO systems when uplink/downlink channel reciprocity does not hold (e.g., in frequency division duplexing systems). From a coding viewpoint, the DL channel state acquired at the users via DL training can be seen as an information source that must be conveyed to the base station via the UL communication channels. The transmission of a source through a channel can be accomplished either by separate or joint source-channel coding (SSCC or JSCC). In this work, using classical remote distortion-rate (DR) theory, we first provide a theoretical lower bound on the channel estimation mean-square-error (MSE) of both JSCC and SSCC-based feedback schemes, which however requires encoding of large blocks of successive channel states and thus cannot be used in practice since it would incur in an extremely large feedback delay. We then focus on the relevant case of minimal (one slot) feedback delay
     and propose a practical JSCC-based feedback scheme that fully exploits the channel second-order statistics to optimize the dimension projection in the eigenspace. We analyze the large SNR behavior of the proposed JSCC-based scheme in terms of the quality scaling exponent (QSE). Given the second-order statistics of channel estimation of any feedback scheme, we further derive the closed-form of the lower bound to the ergodic sum-rate for DL data transmission under maximum ratio transmission and zero-forcing precoding.    
     Via extensive numerical results, we show that our proposed JSCC-based scheme outperforms known JSCC, SSCC baseline and {deep learning-based} schemes and is able to approach the performance of the optimal DR scheme in the range of practical SNR.         
\end{abstract}


\section{Introduction}
\label{sec:intro}
Accurate and timely channel state information at the transmitter (CSIT) is crucial for massive MIMO systems to achieve the benefits of the large dimensional antenna array, such as high beamforming gain and large spectral efficiency \cite{caire2010multiuser, jindal2006mimo}. The uplink (UL) CSI can be estimated at the base station (BS) by receiving the pilot signals from the users in UL. However, the CSI acquisition in downlink (DL) is generally more difficult to obtain, especially in systems where
UL/DL channel reciprocity does not hold. Examples of such systems are frequency division duplexing (FDD) systems, systems with non-reciprocal/calibrated radios, and systems such as IEEE 802.11ax, where the UL and DL scheduler make independent and possibly random/access decisions, such that it is not guaranteed that each DL frame is preceded by a UL frame carrying timely CSI. For such systems, the DL CSI at the BS is obtained by explicit DL probing (via pilots) and closed-loop feedback.

\subsection{Related Work}
In the literature, most works consider codebook-based quantized CSI feedback where the CSI is encoded into bits by the users and transmitted back and decoded at the BS \cite{jindal2006mimo, jiang2015achievable}.  We refer to this ``classical codebook-based feedback'' separated source-channel coding feedback (SSCC-feedback) scheme, which is also used in the current 3rd Generation Partnership Project (3GPP) standard and industry release, see e.g., \cite{3gpp2020technical,huawei2019discussion}. Nevertheless, these works avoid discussing the difficulty of SSCC-feedback with respect to channel coding, which is essential since the codebook-based feedback scheme requires sending bits through the UL, {which must be received with high reliability and low delay before the channel changes significantly}. As an alternative to SSCC-feedback, the joint source-channel coding feedback (JSCC-feedback) scheme (also known as the ``analog'' feedback that was first introduced in \cite{fast2006marzetta} and recently studied in \cite{khalilsarai2022channel, khalilsarai2023fdd}) is based on the direct mapping of the received training signal onto a UL modulated signal and direct mapping of the noisy UL signal into the channel estimate at the BS, which can be implemented directly in the ``low physical layer'' without the need of piggybacking the CSI bits into a UL codeword \cite{samardzija2006unquantized, thomas2005obtaining}, resulting in much lower complexity and delay for the JSCC feedback. 

{
A particular challenge of CSI feedback in modern wideband massive MIMO systems is the prohibitively large channel dimension due to the large antenna array and large number of subcarriers. In order to guarantee the UL data transmission rate, the DL CSI feedback signal should not cost too much in terms of the UL channel resources. This typically requires some form of compression (aka, lossy source coding). Designing the feedback strategies under limited feedback resources has received significant attention in recent years, where the proposed solutions typically depend on the deep learning technology (e.g., \cite{wang2021compressive}) to overcome the high computational complexity of traditional CSI feedback methods, such as codebook-based (e.g., \cite{jindal2006mimo}) and compression sensing-based (e.g., \cite{han2017compressed}) methods, see more references in \cite{guo2022overview}. 
A deep neural network (DNN)-based JSCC feedback scheme was first proposed in \cite{mashhadi2020cnn} using a convolutional neural network (CNN). 
{More recently, a state-of-the-art deep learning-based approach has incorporated attention mechanisms in both the encoder and decoder functions \cite{cui2022transnet}. }
A more general deep JSCC (DJSCC) framework with signal-to-noise ratio (SNR) adaption for CSI feedback is proposed in \cite{xu2022deep}. In \cite{jang2019deep}, JSCC combined with DNN is also considered to solve the channel aging problem due to feedback delay.
Nevertheless, the DNN-based approaches generally have severe disadvantages and limitations. First, the DNN-based works often assume perfect CSI at the user for compression and feedback, which is an unrealistic assumption in massive MIMO systems with large channel dimensions. Second, DNN models are typically trained on limited environmental scenarios and they are highly dependent on training data sets, resulting in weak generalization ability. 
In addition, DNN-based works do not explicitly incorporate second-order statistics.
To the best of our knowledge, the only DNN-based work that exploits the second-order statistics of the channel is our recent work \cite{song2023deep}, where it has been shown that even the Linear JSCC scheme outperforms the seminal deep-learning-based work in \cite{sohrabi2021deep} when second-order statistics are exploited only at the BS side.
% In our work we further provide a clean and rigorous lower bound on the channel estimation error. Our comparison with this information theoretic limit is already sufficient to assess the goodness of the proposed scheme, which is superior to any DNN-based scheme proposed so far in terms of analytical tractability, full generality, clean insight on the role of the scheme parameters (in particular, DL training and UL feedback block length), low delay, and low complexity.
% Therefore, in view of the aforementioned disadvantages/limitations of the deep-learning-based methods, and given the information theoretical limit, a comprehensive comparison with numerous deep learning models becomes unnecessary.
 
}

\subsection{Contribution}
In this paper, we consider the most general CSI feedback problem in multi-user, multicarrier massive MIMO systems, where DL common training, UL feedback, and channel recovery are jointly considered.
For a rigorous analysis of the studied CSI feedback problem, we first provide a lower bound on the mean square error (MSE) of the channel estimate to any JSCC or SSCC feedback scheme by applying the results of the remote distortion rate theory \cite{eswaran2019remote}. Then, for the interesting case with compressed feedback dimension, we propose a JSCC-feedback scheme that exploits the knowledge of the channel second-order statistics at both the BS and users\footnote{It should be noted that the great majority of works considering MMSE channel estimation assume that the channel statistics, i.e., the channel covariance matrix for zero-mean Gaussian (Rayleigh fading) channels with antenna correlation, is known at the estimator (e.g., see {\cite{bjornson2020scalable}}). {A more detailed elaboration of the availability of channel statistics is given in Remark~\ref{remark:second-order}.}}
and derive its asymptotic performance in terms of the MSE decay in high SNR.
JSCC is generally better than SSCC in terms of end-to-end distortion error (in our case the MSE of channel estimate) within a practical range of SNR, especially for low-delay short-block communication, which is also verified in our numerical results.
% \subsection{Contributions}

The main contributions are listed as follows:
\begin{itemize}
    \item From the perspective of information theory, we formulate the feedback problem under the JSCC and SSCC schemes. We provide a distortion-rate (DR) based result that gives a lower bound of the CSI estimation error to any (JSCC and SSCC) feedback scheme.
    % It is intuitive that the result of DR provides also a lower bound of CSI estimation error to JSCC schemes due to the source-channel separation theorem, which is formally verified in Lemma~\ref{lemma: AF_DF_RD}. 
     % \red{We should emphasize that the DR result is not necessarily the lower bound to JSCC scheme.  However, our result in Lemma 1 shows that the DR actually gives a lower bound to JSCC scheme. Giuseppe also questioned on this point, he asked in his comment: '' The DR scheme gives the theoretical lower bound of any achievable feedback scheme''}
    \item {Focusing on compressed feedback dimension and minimal feedback delay (i.e., block-by-block feedback), we propose the truncated Karhunen-Lo\`eve (TKL) scheme, which is a simple JSCC feedback scheme that exploits the second-order statistics of the channel by sequentially truncating the signal in the eigenspace under the Karhunen-Lo\`eve (KL) expansion and optimizing the power allocation to minimize the estimation error using a water filling solution.} We further provide the quality scaling exponent (QSE) of the proposed JSCC feedback scheme to analyze the large SNR behavior of the resulting channel estimate MSE ({{Theorem}}~\ref{theorem:QSE}).  
    \item To easily evaluate the performance of the proposed feedback scheme in terms of the practically considered DL ergodic sum rate, we derive closed-form expressions of the use-and-then-forget (UatF) lower bound of the ergodic rate under both commonly used maximum ratio transmission (MRT) and zero-forcing (ZF) precoding based on the obtained channel estimates of different feedback schemes (Lemma~\ref{lemma: MRT} and \ref{lemma: ZF}).
\end{itemize}
% \red{The extension numerical results shows ...}
Via extensive numerical results, we first demonstrate the validity of the QSE derivation of the proposed JSCC scheme. We then show that under the interesting case with practical SNR ($-10$ dB to $10$ dB) and limited feedback resources, the DL achievable ergodic rate of the proposed JSCC scheme not only outperforms the results of practical SSCC scheme but also approaches the information-theoretic bound (DR scheme). 
{In addition, the proposed JSCC scheme outperforms the deep learning-based scheme under the idealized assumption.}

\section{System Model}
\begin{figure}[!t]
        \centering
        \includegraphics[width=1\columnwidth]{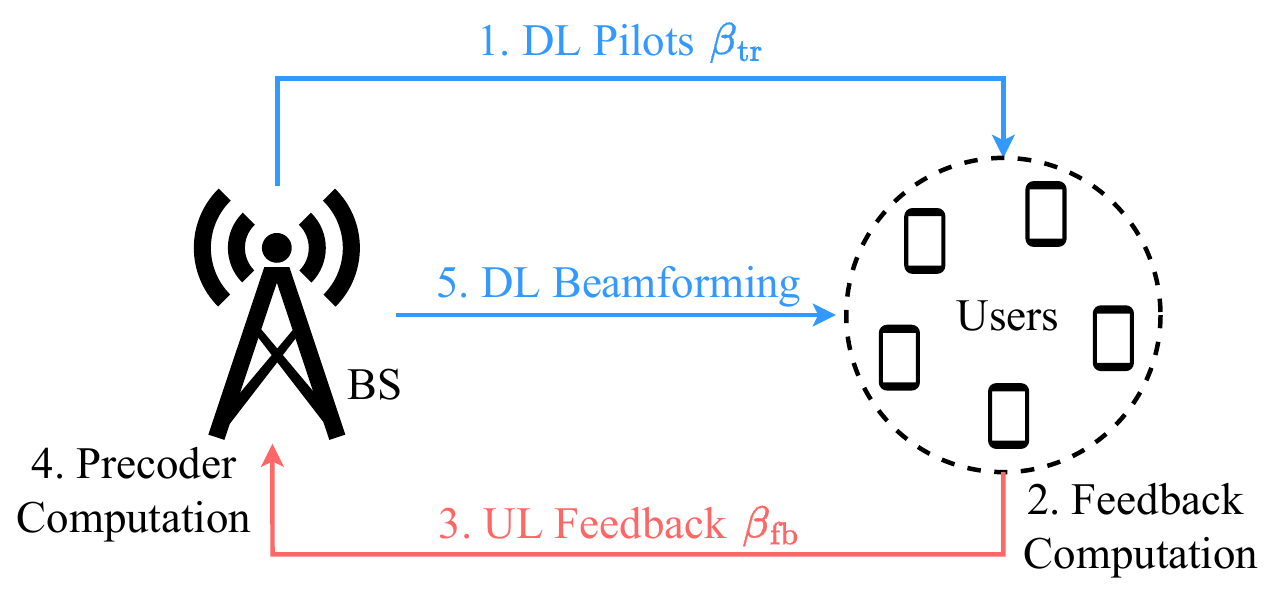}
        \caption{{Schematic for closed-loop FDD DL CSIT acquisition with $\betatr$ training dimension and $\betafb$ feedback dimension.}}
        \label{fig:fdd_traning}
        \vspace{-6mm}
\end{figure}
We consider a multicarrier massive MIMO system, where a BS equipped with $M$ antennas operates in OFDM mode over $N$ subcarriers and serves $K \leq M$ single-antenna users simultaneously. We assume that UL/DL channel reciprocity cannot be exploited (see motivation in Section~\ref{sec:intro}) and that the CSI at the BS for DL precoding must be acquired in a closed loop, via DL common training and CSI feedback from the users to the BS, {see Fig.~\ref{fig:fdd_traning} for the intuitive view of the considered DL CSIT acquisition problem}.  We also assume the block-fading channel model where the channel state information is constant over a single frame with $T$ OFDM symbols (corresponding to the channel coherence time) and changes frame by frame according to an i.i.d. process.
% \footnote{{In typical wireless/mobile communications operating outdoor, with the carrier frequency range between $2$ and $6$ GHz and user mobility up to a few tens of km/h, $T$ may vary from a few hundred to a few thousand of symbols. For practical reasons, though, actual systems perform channel estimation on much shorter blocks (the so-called resource blocks) of \red{$12 \times 14 = 168$} symbols, specified in standards such as 4G-LTE and 5GNR \cite{dahlman20134g, hui2018channel, krasniqi2018performance}.}}
During the DL common training phase, the BS broadcasts a pilot sequence over $T_p$ out of  $T$ symbols and over a subset of $N_p \leq N$ subcarriers to all users, which results in a total pilot dimension of $\betatr = T_p N_p$. The set of subcarrier indices for transmitting pilots is denoted as $\Nc_p \subseteq \Nc = [N]$.\footnote{For a positive integer $N$, we define $[N] \triangleq \{1, \dots, N\} $.
}  The received noisy pilot signal at the $k$-th user over subcarrier $n$ is denoted as $\yv_k^{\rm tr}[n]$ and given by
\begin{equation}
    \yv_k^{\rm tr}[n] = \Xm[n] \hv_k[n] + \nv_k^{\rm tr}[n], ~\forall k \in \Kc, \, n\in\Nc_p,
\end{equation}
where $\Kc = [K]$ and $\hv_k[n] \in \CC^M$ is DL channel vector of the $k$-th user at the subcarrier $n$, and where $\Xm[n]\in\bC^{T_p \times M}$ is the pilot matrix at the subcarrier $n$, whose each element follows complex Gaussian distribution, i.e., $\Xm_{i, j}[n] \sim \Cc\Nc(0, \frac{\mathsf{snr}_{\rm dl}}{M})$, $\forall i \in [T_p], ~\forall j \in [M]$, where $\mathsf{snr}_{\rm dl}$ is the DL transmit power that coincides with DL ``pre-beamforming'' SNR, under the normalization of the noise variance such that the DL additive white Gaussian noise (AWGN) is $\nv^{\rm tr}_k[n] \sim \mathcal{CN}(\mathbf{0}, \mathbf{I})$.

Considering all $\betatr$ training measurements, the received pilot signal at the $k$-th user $\yv_k^{\rm tr} = [\yv_k^{\rm tr}[n_1]^\transp, \dots, \yv_k^{\rm tr}[n_{N_p}]^\transp]^\transp \in \bC^{\betatr}$ can be written as\footnote{Note that we use $\{n_1, n_2,\dots, n_{N_p}\}$ to indicate the subcarrier indices of the probed subcarriers.}
\begin{equation}
    \yv_k^{\rm tr} = \Xm \hv_k + \nv_k^{\rm tr}, ~\forall k \in \Kc,
\end{equation}
where $\hv_k := [\hv_k[1]^\transp, \dots, \hv_k[N]^\transp]^\transp \in \bC^{MN} \sim \mathcal{CN}(\mathbf{0}, \Cm_{\hv_k})$ with its covariance matrix $\Cm_{\hv_k} = \bE[\hv_k\hv_k^\herm]$, and where $\Xm \in \bC^{\betatr \times M N}$ consisting of all pilot matrices of $N_p$ subcarriers. Specifically, $\Xm = (\boldsymbol{\chi}_{\ell,n})^{n\in[N]}_{\ell\in[N_p]}$, where $\boldsymbol{\chi}_{\ell,n} = \Xm[n_{\ell}]$ if $n_\ell \in \Nc_p$ and $\boldsymbol{\chi}_{\ell,n}=\mathbf{0}$ otherwise.  
An intuitive illustration of the structure of $\Xm$ is given as  
\begin{equation}
    \Xm  = \begin{bmatrix}
                \cdots & \Xm[n_1] & \cdots   & \cdots   & \cdots & \cdots & \cdots\\
                \cdots & \cdots   & \cdots   & \Xm[n_2] & \cdots & \cdots & \cdots\\
                \vdots & \vdots   & \vdots   & \vdots   & \vdots & \vdots & \vdots \\
                \cdots & \cdots   & \cdots   & \cdots   & \cdots & \Xm[n_{N_p}] & \cdots  \\
            \end{bmatrix},
\end{equation}
where all blocks other than $\{\Xm[n_1],\dots,\Xm[n_{N_p}]\}$ are $\mathbf{0}$. 
 
%%%%%%%%%%%%%%%%%%% UL Channel
Upon receiving the pilots, the user computes a feedback message containing information about the DL CSI, and sends it back to the BS via $\beta_{\rm fb}$ UL channel uses, i.e., time-frequency symbols. We consider the UL channel as the MIMO multiple-access channel (MAC), where all users send their feedback to the BS simultaneously under the same frequency band. We assume that the BS has perfect knowledge of the UL CSI, since the UL CSI can be estimated from UL pilots inserted in each UL slot. We denote the user transmit power as $\mathsf{snr}_{\rm ul}$ and note that $\mathsf{snr}_{\rm ul}$ corresponds also the UL SNR under the normalized UL noise power. 
Since normally the transmit power of the user is much less than that of the BS, i.e., $\mathsf{snr}_{\rm ul} < \mathsf{snr}_{\rm dl}$, in this work, we assume that the transmit power of the BS is equal to the total transmit power of all users, i.e., $\mathsf{snr}_{\rm dl} = K\mathsf{snr}_{\rm ul}$.
{Then, under the UL MIMO-MAC assumption and UL SNR defined above as well as a diversity-multiplexing trade-off factor (pre-log factor) of one \cite{caire2010multiuser}, the capacity of the simplified conceptual point-to-point UL channel per user is given as
\begin{equation}\label{eq:C_ul}
    C_{\rm ul} = \log (1 + \kappa M \mathsf{snr}_{\rm ul}) = \log (1 +  \kappa M \mathsf{snr}_{\rm dl} / K),
\end{equation}
where $\kappa \in [0,1]$ is the multiuser efficiency that plays as a penalty factor incurred by the presence of the other users in the MIMO-MAC with respect to the ideal performance of the full beam-forming gain \cite{verdu1998multiuser}. The exact value of $\kappa$ depends on the detector and can be computed depending on the channel statistics and the large-system limits as in the standard large-system MIMO analysis (e.g., see \cite{marzetta2010noncooperative, huh2012achieving}).  For instance, under zero-forcing detection, $\kappa=1- K/M$ 
 \cite{verdu1999spectral}. }

%%%%%%%%%%%%%%%%%%% Feedback scheme
After receiving the pilot measurements, the user has two options for feeding back information to the BS, namely the JSCC feedback and DSCC feedback. 
\subsection{JSCC Feedback Scheme}
\begin{figure}[!ht]
        \centering
        \begin{tikzpicture}[node distance = 0.01\textwidth, auto]
        \tikzstyle{neuron} = [circle, draw=black, fill=white, minimum height=0.02\textwidth, inner sep=0pt]
        \tikzstyle{rect} = [rectangle, minimum width=0.02\textwidth, minimum height=0.02\textwidth, draw=black, fill=white]
            % input data
            \node [rect] (x){\small Observation};
            \node [below of=x, yshift=-0.02\textwidth] (x1){$\yv_k^{\rm tr}$};
            \node [left of=x, xshift=0.12\textwidth, rect] (z) {\small Encoder};
             \node [below of=z, yshift=-0.02\textwidth] (z1) {$f_{\rm enc}(\cdot; \Cm_{\hv_k})$};
             \node [left of=z, xshift=0.1\textwidth, neuron] (sample) {+};
                 \node [below of=sample,yshift=-0.05\textwidth] (noise) {$ \nv_k^{\rm ul} \sim \mathcal \Cc\Nc(\mathbf{0}, \Id)$};
         \node [left of=sample, xshift=0.09\textwidth, rect] (dec) {\small Decoder};
         \node [below of=dec, yshift=-0.02\textwidth] (dec1) {$f_{\rm dec}(\cdot; \Cm_{\hv_k})$} ;
                  \node [left of=dec, xshift=0.12\textwidth, rect] (recon) {\small Estimate CSI};
                \node [below of=recon, yshift=-0.02\textwidth] (recon1) {$\widehat{\hv}_k^{\rm jscc}$} ;
                   \draw [->,line width=1pt] (x) -- (z);
                     \draw [->,line width=1pt] (z) -- node[above]{$\zv_k$}(sample);
                       \draw [->,line width=1pt] (sample) --node[above]{$\yv_k^{\rm fb}$} (dec);
                        \draw [->,line width=1pt] (dec) -- (recon);
                          \draw [->,line width=1pt] (noise) -- (sample);
                        %\draw [->,line width=1pt] (R2) -- node[above] {$Z_2$}(D)
        \end{tikzpicture}  
        \caption{The block diagram of the JSCC feedback scheme.}
        % \vspace{-2mm}
        \label{system_block_AF}
\end{figure}
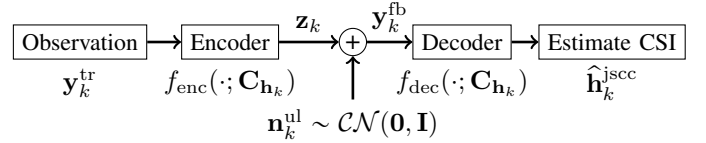

In the JSCC feedback scheme, as shown in Fig. \ref{system_block_AF}, assuming the availability of channel second-order statistics at both BS and user sides, the user applies a direct mapping via a deterministic function to obtain the feedback codewords, and then the BS estimates the DL CSI from the noisy UL feedback. Specifically, each user obtains the feedback codewords as quadrature amplitude modulation (QAM) symbols $\zv_k \in\bC^{\beta_{\rm fb}} $ with unquantized I and Q components\footnote{Modern communication systems make use of very large QAM alphabets, e.g., $4096$ QAM ($802.11$ax). This corresponds to quantizing the I and Q components of the signal in the frequency domain with 6 bits. {This corresponds to a $6 \times 6 = 36 ~$dB signal to quantization noise per component (6 dB per bit scaling), which is at all effects equivalent to unquantized symbol transmission.}} by applying a deterministic mapping as the encoder $f_{\rm enc}(\cdot; \Cm_{\hv_k}): \mathbb{C}^{\beta_{\rm tr}} \rightarrow \mathbb{C}^{\beta_{\rm fb}} $ from the pilot dimension $\beta_{\rm tr}$ to the feedback dimension $\beta_{\rm fb}$, which is given as
\begin{align}
    \zv_k = f_{\rm enc}(\yv_k^{\rm tr}; \Cm_{\hv_k}).
\end{align}

The received noisy feedback signal at the BS is given by 
\begin{align}\label{eq: feedback_AF}
    \yv_k^{\rm fb} = \zv_k + \nv_k^{\rm ul},
\end{align}
where $\nv^{\rm ul}_k \sim \Cc\Nc(\mathbf{0}, \Id)$ is the UL AWGN with normalized noise power, and the feedback symbols should satisfy the transmit power constraint due to the UL MIMO-MAC channel
\begin{align}\label{eq:fb_power}
    \mathbb{E}[\| \zv_k\|^2]  \leq P_{\rm ul},
\end{align}
where $P_{\rm ul} \triangleq \betafb \kappa M \snrul$ is the effective UL transmit power.

After receiving the feedback signal $\yv_k^{\rm fb}$, the BS obtains the DL channel estimate as 
\begin{align}
    \widehat{\hv}_k^{\rm jscc} = f_{\rm dec}(\yv_k^{\rm fb}; \Cm_{\hv_k}),
\end{align}
where the decoder function $f_{\rm dec}(\cdot; \Cm_{\hv_k})$ transforms the received feedback symbols with dimension $\beta_{\rm fb}$ to the channel estimate with dimension $MN$ based on the second-order statistics of the DL channels.

Overall, we aim to minimize the MSE of the channel estimate by jointly designing the encoder and decoder functions with the knowledge of the channel covariance at both the user and the BS.\footnote{It can be shown that minimizing the MSE of the estimated channel maximizes the use-and-then-forget (UatF) rate bound. Therefore, the channel estimation MSE is precisely the correct metric to assess the quality of the CSI estimation/feedback loop. 
} Hence, the optimization problem of the JSCC scheme for each user is formulated as   
\begin{subequations}\label{eq: AF}
    \begin{align}
        \underset{{f_{\rm enc}(\cdot;  \cdot), \; f_{\rm dec}(\cdot; \cdot)}}{\text{minimize}} \; \;\;\quad \quad&  \bE\left[\|\hv_k - \widehat{\hv}_k^{\rm jscc}\|^2\right] \\
 	 \text{subject to} \;\;\; \quad \qquad&    \eqref{eq:fb_power}.  
     \end{align}
    \end{subequations}

\subsection{SSCC Feedback Scheme}
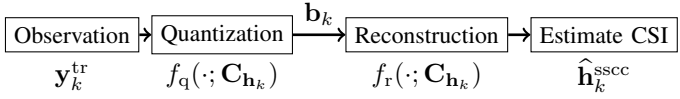
\begin{figure}[!ht]
        \centering
        \begin{tikzpicture}[node distance = 0.01\textwidth, auto]
        \tikzstyle{neuron} = [circle, draw=black, fill=white, minimum height=0.02\textwidth, inner sep=0pt]
        \tikzstyle{rect} = [rectangle, minimum width=0.02\textwidth, minimum height=0.02\textwidth, draw=black, fill=white]
            % input data
            \node [rect] (x){\small Observation};
            \node [below of=x, yshift=-0.02\textwidth] (x1){$\yv_k^{\rm tr}$};
            \node [left of=x, xshift=0.12\textwidth, rect] (z) {\small Quantization};
             \node [below of=z, yshift=-0.02\textwidth] (z1) {$f_{\rm q}(\cdot; \Cm_{\hv_k})$};
         \node [left of=z, xshift=0.16\textwidth, rect] (dec) {\small Reconstruction};
         \node [below of=dec, yshift=-0.02\textwidth] (dec1) {$f_{\rm r}(\cdot; \Cm_{\hv_k})$} ;
                  \node [left of=dec, xshift=0.14\textwidth, rect] (recon) {\small Estimate CSI};
                \node [below of=recon, yshift=-0.02\textwidth] (recon1) {$\widehat{\hv}_k^{\rm sscc}$} ;
                   \draw [->,line width=1pt] (x) -- (z);
                     \draw [->,line width=1pt] (z) -- node[above]{${\bv}_k$}(dec);
                        \draw [->,line width=1pt] (dec) -- (recon);
        \end{tikzpicture}  
        \caption{The block diagram of the SSCC feedback scheme.}
          % \vspace{-2mm}
        \label{system_block_DF}
\end{figure}
In the SSCC feedback scheme, as shown in Fig. \ref{system_block_DF}, the user first quantizes the received noisy pilot signal into binary sequences, and then sends the bits with UL channel coding, and finally the BS decodes the DL channel based on the quantization bits. 
Same as the JSCC case, we also assume that the channel second-order statistics are known at both the user and BS. Given the number of bits for encoding as $R_k$, the $k$-th user generates the quantization sequences ${\bv}_k \in \{0, 1\}^{2^{R_k}} $ via the quantization function $f_{\rm q}(\cdot;\cdot): \mathbb{C}^{\beta_{\rm tr}} \rightarrow \{0, 1\}^{2^{R_k}}$ based on the DL channel covariance matrix, which is given as
\begin{align}
    {\bv}_k &= f_{\rm q}(\yv_k^{\rm tr}; \Cm_{\hv_k}).
\end{align}
We assume that the quantization sequences can be transmitted error-free from the user to the BS through UL channels. 
Upon receiving the bit stream, the DL channel estimate can be obtained via the reconstruction function $f_{\rm r}(\cdot; \cdot)$ as 
\begin{align}
    \widehat{\hv}_k^{\rm sscc} &= f_{\rm r}({\bv}_k; \Cm_{\hv_k}), 
\end{align}
where $f_{\rm r}(\cdot; \cdot)$ transforms the binary sequences of dimension $2^{R_k}$ into the vector of dimension $MN$ based on the second-order statistics of the DL channels.
Again, we aim to minimize the MSE of the channel estimation by jointly designing the quantization and de-quantization functions with the knowledge of the channel covariance at both the user and the BS. Hence, the optimization problem of the SSCC scheme for each user is formulated as   
\begin{subequations}\label{eq: DF}
    \begin{align}
        \underset{{f_{\rm q}(\cdot; \cdot), \; f_{\rm r}(\cdot; \cdot)}}{\text{minimize}} \; \;\;\quad \quad&  \bE\left[\|\hv_k - \widehat{\hv}_k^{\rm sscc}\|^2\right] \\
 	 \text{subject to}  \quad \qquad&   R_k \leq R,   
     \end{align}
    \end{subequations}
where $R$ denotes the rate constraint in UL channels.

%%%%%%%%%%%%%%%%%%%%%%%%

\subsection{Relation between the UL Channels in JSCC and SSCC}
We discuss the choice of rate constraint $R$ to relate and compare the JSCC and the SSCC feedback schemes in a fair manner. In the JSCC, $\beta_{\rm fb}$ channel uses of the UL channel are allocated for feeding back the codewords. In order to make a fair comparison between JSCC and SSCC, the rate constraint in the SSCC feedback scheme is defined as $R = \beta_{\rm fb} C_{\rm ul}$, where $C_{\rm ul}$ is given in \eqref{eq:C_ul}.

\begin{remark}
    Note that from the latency perspective, the JSCC scheme sends the CSI estimates back to BS immediately, as soon as they are obtained by the UE, and \textbf{does not} require channel coding and piggybacking in the UL payload. In contrast, the SSCC scheme \textbf{does} require channel coding to send digital information without bit errors, which involves non-trivial coding and decoding delays for piggybacking the CSI quantization bits into a UL-coded payload and decoding at the BS. We emphasize that our comparison is very generous with the SSCC feedback since it is assumed that the quantization bits can be transmitted at a rate equal to the UL channel capacity and error-free without coding delay.  Therefore, our results are somehow already biased in favor of the SSCC feedback scheme.  \hfill $\lozenge$
\end{remark}

\section{Lower Bound to the CSI Estimation Error: the Distortion-Rate (DR) Scheme}
\label{ch: optimal_bound}
The problem in \eqref{eq: DF} is a long-standing and unsolved problem due to the difficulty in designing quantization and reconstruction functions for finite block-length. A lower bound on the achievable distortion can be obtained by relaxing the problem and considering the compression/reconstruction of an arbitrarily large sequence of channel vectors. We assume that the channel is a stationary sequence of i.i.d. vector symbols $\lim_{n \rightarrow \infty}\{\hv_k^{(i)}\}_{i=1}^n$, where $\hv_k^{(i)}$ represents the $i$-th channel realization and $n$ denotes the source encoding block length. Let ${\yv_k^{\rm tr}}^{(i)}$ denote the DL measurements on channel $\hv_k^{(i)}$. Our focus is on the design of the quantization function $f_{\rm q}^n(\cdot)$, which maps the $n$-sequence DL training signals into binary sequences, and the reconstruction function $f^n_{\rm r}(\cdot)$, which decodes the binary streams into sequences of channel estimates $\{\widehat{\hv}_k^{(i)}\}_{i=1}^n$. The problem then becomes 
\begin{subequations}\label{eq: DF_LB}
    \begin{align}
        \underset{f^n_{\rm q}(\cdot), f^n_{\rm r}(\cdot)}{\text{minimize}} \quad &\lim_{n \rightarrow \infty} \frac{1}{n} \sum_{i=1}^n \bE \left[\|\hv_k^{(i)} - \widehat{\hv}_k^{(i)}\|^2\right] \\
 	    \text{subject to} \quad &\lim_{n \rightarrow \infty} \frac{1}{n} H\left( f_{\rm q}^n({\yv_k^{\rm tr}}^{(1)}, \cdots,{\yv_k^{\rm tr}}^{(n)})\right) \leq R,    
    \end{align}
\end{subequations}
which gives a lower bound to \eqref{eq: AF} and \eqref{eq: DF} due to the relaxation to the infinite block-length.
Information theory \cite{gamal2012network} proves the existence of quantization and reconstruction functions for sufficiently large $n$ achieving the minimum possible MSE, namely, the remote distortion rate function $D_k^{\rm r}(R)$ given by the solution of the following problem: 
\begin{subequations}\label{eq: DF_LB_2}
    \begin{align}
        \underset{p({\widehat{\hv}_k| \hv_k})}{\text{minimize}} \quad &\bE \left[\|\hv_k - \widehat{\hv}_k\|^2\right] \label{eq:DF_LB_obj} \\
 	    \text{subject to} \quad &I(\yv^{\rm tr}_k; \widehat{\hv}_k) \leq R, 
    \end{align}
\end{subequations}
where $p({\widehat{\hv}_k| \hv_k})$ is the conditional probability density function of $\widehat{\hv}_k$ given $\hv_k$.
% We show in the following lemma that \eqref{eq: DF_LB_2} is a lower bound on the MSE of channel estimation of both JSCC and SSCC schemes.
% \begin{lemma}\label{lemma: AF_DF_RD}
%     When $R = \beta_{\rm fb} C_{\rm ul}$, the MSE in \eqref{eq:DF_LB_obj} is a lower bound both to that of the JSCC feedback scheme in \eqref{eq: AF} and that of the SSCC feedback scheme in \eqref{eq: DF}.
% \end{lemma}
% \begin{proof}
%     See Appendix \ref{sec:lemma_RD}.
% \end{proof}
It can be shown that, based on \cite[(42) in Appendix A]{eswaran2019remote}, the remote distortion-rate function in our case can be divided into two components: 
\begin{equation}
    D_k^{\rm r}(R) = D^{\rm mmse}_k + D_k(R),
\end{equation}
where $D^{\rm mmse}_k$ is the estimation error of the MMSE estimation at the user side, and where 
\begin{equation}
    D_k(R) \triangleq \min_{p(\widehat{\hv}_k| \uv_k): I(\uv_k; \widehat{\hv}_k) \leq R} \mathbb{E}[\|\uv_k - \widehat{\hv}_k\|^2]
\end{equation}
is the distortion-rate function for the MMSE estimate $\uv_k$. Specifically, the user can first obtain the MMSE estimation of DL CSI from the noisy pilot observation, given as
\begin{align}
    \uv_k &= \mathbb{E} [\hv_k|\yv_k^{\rm tr}] = \Cm_{\hv_k} \Xm^{\herm} (\Xm \Cm_{\hv_k} \Xm^{\herm} + \Id_{\beta_{\rm tr}})^{-1} \yv_k^{\rm tr},
\end{align}
where the covariance matrix of $\uv_k$ is given by 
\begin{align}
    \Cm_{\uv_k}  = \Cm_{\hv_k} \Xm^{\herm} (\Xm \Cm_{\hv_k} \Xm^{\herm} + \Id_{\beta_{\rm tr}})^{-1} \Xm \Cm_{\hv_k},
\end{align}
and the estimation error $D^{\rm mmse}_k$  is obtained as $D^{\rm mmse}_k = \tr(\Cm_{\hv_k}-\Cm_{\uv_k})$. 
Next, denoting the $i$-th eigenvalue of $\Cm_{\uv_k}$ as $\lambda_{i}^{\uv_k}$, we can derive a closed-form solution to the remote distortion-rate function $D_k^{\rm r}(R)$ in the following lemma.  
\begin{lemma}[remote distortion-rate function]
    For a fixed pilot matrix $\Xm$ and a given rate $R$, the remote distortion-rate function is given by 
    \begin{align}
        D_k^{\rm r}(R) = D_k^{\rm mmse} + \sum_{i=1}^{MN} \min\{\gamma, \lambda_{i}^{\uv_k}\},
    \end{align}
    where $\gamma$ is chosen to make full use of the UL rate constraint $R$, i.e.,\footnote{$[x]_+$ denotes the function $\max(x, 0)$.}
    \begin{align} \label{eq: gamma}
        \sum_{i=1}^{MN} \left[\log_2\left(\frac{\lambda_i^{\uv_k}}{\gamma}\right)\right]_+  \triangleq  R.
    \end{align}
    Similarly, given a target distortion $D$ the remote rate-distortion function $R_k^{\rm r}(D)$ is given by 
    \begin{align}
       R_k^{\rm r}(D) = \sum_{i=1}^{MN} \left[\log_2\left(\frac{\lambda_i^{\uv_k}}{\widetilde{\gamma}}\right)\right]_+, 
    \end{align}
    where $\widetilde{\gamma}$ is chosen such that $D_k^{\rm mmse} + \sum_{i=1}^{MN} \min\{\widetilde{\gamma}, \lambda_{i}^{\uv_k}\} \triangleq D$.
\end{lemma}
\begin{proof}
   The proof is quite classical. Readers can refer to e.g., \cite[Lemma 1]{khalilsarai2023fdd}.  
\end{proof}

\subsection{Channel Estimate for the DR Scheme}
For numerically evaluating the DR feedback scheme in terms of e.g., precoding design or ergodic rate, it is necessary to have the channel estimates of the DR feedback scheme, denoted as $\widehat{\hv}_k^{\rm dr}$. In the following, we show the way to generate channel estimates $\widehat{\hv}_k^{\rm dr}$. For a fixed pilot matrix $\Xm$ and a given rate constraint $R$, the threshold parameter $\gamma$ can be obtained based on \eqref{eq: gamma}, and thus the distortion vector is given as $\dv_k =[\min\{\gamma,\lambda_1^{\uv_k}\},\dots, \min\{\gamma,\lambda_{MN}^{\uv_k}\}]^\transp$. Then, the channel estimate at user $\uv_k$ can be written as \cite[Theorem 10.3.2]{Cover2006elements}
\begin{align} 
    \uv_k = \widehat{\hv}_k^{\rm dr} + \qv_k,
\end{align}
where $\qv_k \sim \Cc\Nc(\mathbf{0}, \Cm_{\qv_k})$ with $\Cm_{\qv_k}= \Fm_k\diag(\dv_k)\Fm_k^{\herm}$ being the quantization error vector that is independent of the channel estimate $\widehat{\hv}_k^{\rm dr} \sim \Cc\Nc(\mathbf{0}, \Cm_{\widehat{\hv}_k^{\rm dr}})$ with the covariance matrix $\Cm_{\widehat{\hv}_k^{\rm dr}} = \Cm_{{\uv}_k} - \Cm_{\qv_k}$, where $\Fm_k \in \mathbb{C}^{MN \times MN}$ denotes the orthonormal eigenvectors of $\Cm_{\uv_k}$. Therefore, the channel estimate $\widehat{\hv}_k^{\rm dr}$ can be obtained from the MMSE estimation based on $\uv_k$ as 
\begin{align}\label{eq: hat_h_DR}
    \widehat{\hv}_k^{\rm dr} = \bE[\widehat{\hv}_k^{\rm dr} | \uv_k] + \widetilde{\qv}_k, 
    \end{align}
where $\bE[\widehat{\hv}_k^{\rm dr} | \hv_k] = \Cm_{\widehat{\hv}_k^{\rm dr}}\Cm_{{\uv}_k}^{-1} \uv_k $ and $\widetilde{\qv}_k \sim \Cc\Nc(\mathbf{0}, \Cm_{\widetilde{\qv}_k})$ with $\Cm_{\widetilde{\qv}_k} = \Cm_{\widehat{\hv}_k^{\rm dr}}\Cm_{{\uv}_k}^{-1}\Cm_{\widehat{\hv}_k^{\rm dr}}$.

\section{Practical JSCC and SSCC Feedback Schemes}
\label{ch: achievable_schemes}
Although the remote distortion-rate function $D_k^{\rm r}(R)$ is a lower bound for any feedback schemes in \eqref{eq: AF} and \eqref{eq: DF}, it requires encoding blocks of $n$ DL pilot observations together via high-dimensional vector quantizers. 
Beyond its high complexity, the main problem that makes such an approach inapplicable for the CSI feedback is that a block-feedback message can be transmitted only after a very large number $n$ of DL training signals is received. At this point, even though the joint compression and transmission of the block results in a low channel estimation MSE at the BS, these estimates are completely useless since they are stale. As a matter of fact, the CSI estimation and feedback loop must have a low overall delay, such that the BS receives an updated CSI estimate at each coherence interval of $T$ OFDM symbols. 
We refer to such schemes as block-by-block (minimum delay) estimation and feedback schemes.  
Focusing on such schemes, in this section, we study the practical JSCC and SSCC schemes aiming at encoding only {short block-length symbols}. We first introduce the classical SSCC and JSCC schemes, and then propose a novel JSCC scheme.  

\subsection{Entropy-Coded Scalar Quantization (ECSQ) SSCC Scheme}
Similar to the DR scheme, the ECSQ scheme also first applies MMSE estimation to obtain the channel estimates $\uv_k$ and then quantize and encodes the channel estimates in a simpler way. 
Concretely, we can express the estimation $\uv_k$ in its eigenspace by applying Karhunen-Lo\`eve (KL) expansion, which is given by 
\begin{align}\label{eq:w_k}
    \uv_k = \Fm_k \wv_k, 
\end{align}
where each element of $\wv_k$ follows complex Gaussian distribution, i.e., $w_{k, i} \sim \Cc\Nc(0, \lambda_{i}^{\uv_k}), ~\forall~ i \in [MN]$, where $\lambda_{i}^{\uv_k}$ is the $i$-th eigenvalue of $\Cm_{\uv_k}$.
Since we assume that the user and BS both have the perfect knowledge of the channel covariance, $\Fm_k$ is known at the BS side, and thus only the coefficients $\wv_k$ need to be quantized and encoded. The ECSQ scheme encodes the coefficients $\wv_k$ in \eqref{eq:w_k} using scalar quantization followed by entropy encoding based on the results in \cite{ziv1985universal}. Because this scheme is introduced in our previous work, we omit the design details here and the reader can refer to \cite[Section IV.A]{khalilsarai2023fdd}. In the following, we directly give the result of ECSQ.

According to \cite[Theorem 1]{ziv1985universal}, given a certain distortion $D$, the entropy coding of $\wv_k$ after scalar quantization with dithering is not larger than $1.508$ bits per encoded coefficient compared to remote rate-distortion function.
Therefore, the ECSQ scheme requires a rate of\footnote{$\mathbb{1}\{\Ac\}$ denotes the indicator function of the condition $\Ac$, i.e., it gives $1$ when $\Ac$ is satisfied, and $0$ otherwise.}
\begin{align}\label{eq: R_k^ecsq}
    R_k^{\rm ecsq}(D) = R_k^{\rm r}(D) + 1.508 \sum_{i=1}^{MN} \mathbb{1}\left\{\lambda_i^{\uv_k}> \widehat{\gamma}\right\}.
\end{align}
Correspondingly, for a fixed pilot matrix $\Xm$ and a given rate constraint $R$, the minimum achievable channel estimation error in the ECSQ scheme is given by 
\begin{align}
    D_k^{\rm ecsq}(R) = D_k^{\rm mmse} + \sum_{i=1}^{MN} \min\{\widehat{\gamma}, \lambda_{i}^{\uv_k}\}, 
\end{align}
where $\widehat{\gamma}$ is chosen such that the rate constraint of the ECSQ scheme is satisfied, i.e., 
\begin{align}
   \sum_{i=1}^{MN} \left[\log_2\left(\frac{\lambda_i^{\uv_k}}{\widehat{\gamma}}\right)\right]_+ + 1.508 \sum_{i=1}^{MN} \mathbb{1}\left\{\lambda_i^{\uv_k}> \widehat{\gamma}\right\} 
   = R.
\end{align}
The channel estimates $\widehat{\hv}_k^{\rm ecsq}$ based on the ECSQ scheme can be obtained similar to \eqref{eq: hat_h_DR}.

\subsection{Linear JSCC (LJSCC) Scheme}
The classical JSCC feedback applies a simple linear mapping from the noisy pilot observation to the feedback symbols, using the so-called ``spreading'' or ``de-spreading'' matrix $\Wm_k$ that respectively corresponds to the dimension expansion ($\beta_{\rm fb}\geq \beta_{\rm tr}$) or reduction ($\beta_{\rm fb}< \beta_{\rm tr}$), see e.g., \cite[Section IV.B]{khalilsarai2023fdd} \cite[Section III.B]{kobayashi2011training}. Concretely, the codeword produced by the user $k$ after receiving the DL pilot measurements is given by 
\begin{align}\label{eq:LJSCC_code}
    \widetilde{\zv}_k = \Wm_k \yv^{\rm tr}_k, 
\end{align}
where $\Wm_k \in \mathbb{C}^{\beta_{\rm fb} \times \beta_{\rm tr}}$ is a pre-defined full-rank matrix, projecting from the dimension $ \beta_{\rm tr}$ into the dimension $\beta_{\rm fb}$. One possible realization of $\Wm_k$ is a random matrix, where each element follows a complex Gaussian distribution, i.e., $[\Wm_k]_{i, j} \sim \Cc\Nc(0, 1), ~\forall i \in [\beta_{\rm fb}], ~\forall j \in [\beta_{\rm tr}]$. The codeword should satisfy the power constraint in \eqref{eq:fb_power}, which is given as
\begin{align}
    {\widetilde{\zv}}^{\rm norm}_k = \sqrt{\nu_k} \widetilde{\zv}_k, 
\end{align}
where $\nu_k$ is a power scaling factor for user $k$, given as 
\begin{equation}\label{eq:nu}
    \nu_k = \frac{P_{\rm ul}}{\tr\left(\Wm_k(\Xm \Cm_{\hv_k} \Xm^{\herm} + \Id)\Wm_k^{\herm}\right)}.
\end{equation}
Then, the received feedback symbols at the BS is given by 
\begin{align}
    \widetilde{\yv}^{\rm fb}_k &= {\widetilde{\zv}}^{\rm norm}_k + \nv_k^{\rm ul},
\end{align}
which in total is a linear function of $\hv_k$. Hence, the BS can apply the linear MMSE estimation to obtain the channel estimate $\widehat{\hv}_k^{\rm ljscc}$, which is given by 
\begin{align}
    \widehat{\hv}_k^{\rm ljscc} &= \bE[\hv_k | \widetilde{\yv}^{\rm fb}_k]  \\
    &=\Cm_{\hv_k \widetilde{\yv}^{\rm fb}_k} \Cm_{\widetilde{\yv}^{\rm fb}_k}^{-1} \widetilde{\yv}^{\rm fb}_k,
\end{align}
where the corresponding covariance matrices are given as
\begin{align}
    \Cm_{\hv_k \widetilde{\yv}^{\rm fb}_k} &=  \sqrt{\nu} \Cm_{\hv_k} \Xm^{\herm} \Wm_k^{\herm} \\
    \Cm_{\widetilde{\yv}^{\rm fb}_k} &= \nu \Wm_k(\Xm \Cm_{\hv_k} \Xm^{\herm} + \Id_{\beta_{\rm tr}}) \Wm_k^{\herm} + \Id_{\beta_{\rm fb}}.
\end{align}
Finally, the channel estimation MSE of the LJSCC scheme is given as  
    \begin{align}
    D_k^{\rm ljscc}(\beta_{\rm fb}, {\mathsf{snr}}_{\rm ul}) =  \tr(\Cm_{\hv_k} -\Cm_{\widehat{\hv}_k^{\rm ljscc}}), 
\end{align}
where $\Cm_{\widehat{\hv}_k^{\rm ljscc}}$ denotes the covariance matrix of the channel estimates, given as 
\begin{align}
    \Cm_{\widehat{\hv}_k^{\rm ljscc}} &= \Cm_{\hv_k \widetilde{\yv}^{\rm fb}_k} \Cm_{\widetilde{\yv}^{\rm fb}_k}^{-1} \Cm_{\hv_k \widetilde{\yv}^{\rm fb}_k}^{\herm}.
\end{align}

\begin{remark}
    Note that one of the advantages of the LJSCC feedback scheme declared in \cite{khalilsarai2023fdd} is that it does not require the second-order statistics of DL channels at the users. We notice that the LJSCC scheme does require such knowledge, i.e., $\Cm_{\hv_k}$ in our case, to apply power normalization according to \eqref{eq:nu}. However, in practice, given fixed $\Xm$ and $\Wm_k$ (can be predefined and available at both user and BS sides), the power scaling factors $\nu_k$ can be periodically computed at the BS and sent to users, which should be only updated when the channel covariance is significantly changed (corresponding to the channel geometry coherence time which is much longer than the channel vector coherence time). Therefore, the overhead of knowing power scaling factors $\nu_k$ in the LJSCC scheme is much less than the acquisition of the full channel covariance, which still indicates the simplicity of the LJSCC scheme. \hfill $\lozenge$

    % The linear JSCC scheme has slightly exploited the second-order statistics of DL channels to obtain the power normalization factor $\nu$. However, since it is fixed in a geometry coherence time block where the channel covariance matrix remains constant, the BS can compute this fixed value and send it to the users instead of computation at the user side. In a geometry coherence time (which is relatively large), the cost to send this fixed value from the BS to the user can be ignored. 
\end{remark}

\subsection{Truncated Karkune-Lo\`eve (TKL) JSCC Scheme}
Although the LJSCC scheme requires less channel second-order statistics knowledge and simple signal processing, its performance could be worse especially under limited feedback dimension since it treats all dimensions equally. Therefore, in this paper, we propose a novel JSCC feedback that fully exploits the channel second-order statistics.    

In the proposed TKL scheme, we adopt the non-linear encoder function, where the received measurements are first mapped into its eigenspace and then truncated based on the eigenvalues, and finally the power allocation of the truncated signal is optimized to minimize the estimation error. Concretely, we start by finding a proper format of the eigenspace of the covariance matrix of ${\yv_k^{\rm tr}}$. First, we write the singular value decomposition (SVD) of $\Xm \Cm_{\hv_k}^{\frac{1}{2}}$  as 
\begin{align}\label{eq:svd}
    \Xm \Cm_{\hv_k}^{\frac{1}{2}}  = \Um_k \Sigmam_k \Qm_k^\herm,
\end{align}
where the main diagonal of $\Sigmam_k  \in \mathbb{C}^{\beta_{\rm tr} \times MN}$ contains the singular values in the descending order, and where $\Um_k \in \mathbb{C}^{\beta_{\rm tr} \times \beta_{\rm tr}}$ and $\Qm_k  \in \mathbb{C}^{MN \times MN}$ are the corresponding left-singular vectors and right-singular vectors, respectively. Then, with \eqref{eq:svd} the eigenvalue decomposition of the covariance matrix of ${\yv_k^{\rm tr}}$ is derived as
\begin{align}
    \Cm_{\yv_k^{\rm tr}} &= \Um_k \Sigmam_k \Sigmam_k^\herm \Um_k^\herm +\Id_{\beta_{\rm tr}}\nonumber \\
    &= \Um_k (\widebar{\Lambdam}_k + \Id_{\beta_{\rm tr}}) \Um_k^\herm, \label{eq:eigenspace}
\end{align}
where $\widebar{\Lambdam}_k \triangleq \Sigmam_k \Sigmam_k^\herm$.
Given \eqref{eq:eigenspace}, the TKL scheme transforms the received pilot signal into the eigenspaces and truncates at the first $\beta_{\rm fb}$ dominant subspace. Specifically, the feedback symbols can be obtained by  
\begin{align}\label{eq: codeword}
    \widehat{\zv}_k &= \diag(\sqrt{\alphav_k})\Sm_k  \Um_k^\herm \yv^{\rm tr}_k,
\end{align}
where $\Sm_k \in \{0,1\}^{\beta_{\rm fb} \times \betatr}$ is the selection matrix with $\Sm_{i, i} =1, \forall i \in [\beta_{\rm fb}] $ and zeros otherwise,\footnote{Note that the TKL scheme focus on the case with limited feedback dimension, i.e., $\betafb \leq \betatr$ and thus $\Sm_k$ is a fat matrix with $\Sm_k\Sm_k^\herm = \Id_{\betafb}$.} and where $\alphav_k \geq 0$ is the scaling vector to satisfy the power constraint \eqref{eq:fb_power}, given as
\begin{align}
    \bE[\|\widehat{\zv}_k\|_2^2] &=\tr( \diag (\sqrt{\alphav_k}) \Sm_k (\widebar{\Lambdam}_k + \Id_{\beta_{\rm fb}}) \Sm_k^\herm)  \\
    &= \sum_{i=1}^{\beta_{\rm fb}} \alpha_{k, i} (\widebar{\lambda}_{k, i} + 1) \leq P_{\rm ul}. \label{eq:TKL_power}
\end{align}
The received noisy feedback symbols at the BS are given as 
\begin{align}\label{eq: y_fb_TKL}
    \widehat{\yv}_k^{\rm fb} &= \widehat{\zv}_k + \nv_k^{\rm ul} \nonumber \\
    &= \diag(\sqrt{\alphav_k}) \Sm_k \Um_k^\herm (\Xm \hv_k + \nv_k^{\rm tr}) + \nv_k^{\rm ul},
\end{align}
which is still a linear function of $\hv_k$ as the LJSCC scheme. Therefore, the BS can apply linear MMSE estimation on $\widehat{\yv}_k^{\rm fb}$ to obtain the channel estimate $ \widehat{\hv}_k^{\rm tkl}$, which is given by
\begin{align}
    \widehat{\hv}_k^{\rm tkl} &= \bE[\hv_k | \widehat{\yv}^{\rm fb}_k] \nonumber \\
    &=\Cm_{\hv_k \widehat{\yv}^{\rm fb}_k} \Cm_{\widehat{\yv}^{\rm fb}_k}^{-1} \widehat{\yv}^{\rm fb}_k,
\end{align}
where $\Cm_{\hv_k \widehat{\yv}^{\rm fb}_k}$ and $\Cm_{\widehat{\yv}^{\rm fb}_k}$ are given by
\begin{align}
    \Cm_{\hv_k \widehat{\yv}^{\rm fb}_k} &= \Cm_{\hv_k}\Xm^\herm \Um_k \Sm_k^\herm \diag(\sqrt{\alphav_k}) \\ 
    \Cm_{\widehat{\yv}^{\rm fb}_k} &= \diag({\alphav_k})\Sm_k(\widebar{\Lambdam}_k + \Id_{\beta_{\rm tr}})\Sm_k^\herm  +\Id_{\beta_{\rm fb}}.
\end{align}
Accordingly, the channel estimation error in the TKL scheme is given by 
 \begin{align}\label{eq:TKL_MMSE}
    D_k^{\rm tkl}(\beta_{\rm fb}, \mathsf{snr}_{\rm ul}) =  \tr(\Cm_{\hv_k} -\Cm_{\widehat{\hv}_k^{\rm tkl}}), 
\end{align}
where $\Cm_{\widehat{\hv}_k^{\rm tkl}}$ denotes the covariance matrix of the channel estimate $\widehat{\hv}_k^{\rm tkl}$, given as
\begin{align}
    \Cm_{\widehat{\hv}_k^{\rm tkl}} =  \Cm_{\hv_k \widehat{\yv}^{\rm fb}_k} \Cm_{\widehat{\yv}^{\rm fb}_k}^{-1} \Cm_{\hv_k \widehat{\yv}^{\rm fb}_k}^{\herm}.
\end{align}
Now, the power allocation factor $\alphav_k$ is still unknown. We find the optimal power allocation by minimizing the MSE in \eqref{eq:TKL_MMSE} under the power constraint in \eqref{eq:TKL_power}, which is given as
\begin{subequations}\label{eq:opt_ppp}
\begin{align}
    \underset{\alphav_k\geq 0}{\text{minimize}}  \quad &  \tr(\Cm_{\hv_k} -\Cm_{\widehat{\hv}_k^{\rm tkl}}) \\
    \text{subject to} \quad&  \sum_{i=1}^{\beta_{\rm fb}} \alpha_{k, i} (\widebar{\lambda}_{k, i} + 1) \leq P_{\rm ul}. \label{eq:power_constraint}
\end{align}
\end{subequations}
The solution to \eqref{eq:opt_ppp} is cumbersome and thus deferred to Appendix~\ref{sec:lemma_TKL}.

{
\begin{remark}\label{remark:complexity}
    (Comparable complexity of TKL and LJSCC)
{
Since the compressed feedback dimension typically satisfies \(MN \gg \beta_{\rm fb}^2\), the SVD operation in \eqref{eq:svd} dominates the overall computational complexity. Hence, the total complexity of the TKL scheme is \(O(\beta_{\rm tr} (MN)^2)\).} However, the SVD computation depends solely on the channel's second-order statistics, which only need to be updated when these statistics undergo significant changes. The second-order channel statistics remain constant over a much longer period, known as the channel covariance/geometry coherence time, which is typically $100$ to $1000$ times longer than the channel coherence time \cite[Section I]{haghighatshoar2016massive}. 
    This means that the high computational SVD of TKL scheme need to be much less frequently performed compared to CSI estimation, so that their computational cost becomes negligible in long-term periods.
    At each time of CSI estimation, only a linear encoder function in \eqref{eq: codeword} needs to be applied that is with the same complexity order to the encoder function of LJSCC in \eqref{eq:LJSCC_code}.
    As a result, TKL has comparable low complexity to LJSCC on average. \hfill $\lozenge$
\end{remark}

}

{
\begin{remark}\label{remark:second-order}
(Availability of channel statistics) 
The assumption of available channel statistics is commonly made in many published works, not only in the massive MIMO systems, e.g.,  \cite{dietrichpilot2005, shen2017performance, khalilsarai2018fdd, yang2023plug}, but also in the cell-free systems, e.g., \cite{bjornson2020scalable, demir2021foundations, ye2022spectral, miretti2022team, gottsch2022subspace, gottsch2024fairness}. 
The DL channel covariance matrix can be estimated at the BS from a limited number of UL measurements \cite{song2020deep, yang2023structured}, and at the user from DL measurements \cite{ye2022spectral}.
It is noticed from Remark~\ref{remark:complexity} that the second-order channel statistics remain constant over a much longer time compared to the channel coherence time. Therefore, the overhead of estimating the covariance matrix is negligible in the long-term consideration, making the assumption of available channel statistics both feasible and reasonable.
\hfill $\lozenge$
\end{remark}
}

\subsection{Quality Scaling Exponent (QSE)}
% \red{$\alpha$ is used both for power scaling in TKL and QSE. Change one of them to another notation.}

In the literature, the quality scaling exponent (QSE) is commonly used to describe the high SNR behavior of channel estimate MSE. Concretely, the MSE between the true and estimate CSI decreases as $\Theta({\mathsf{snr}}_{\rm dl}^{-\alpha})$ when  ${\mathsf{snr}}_{\rm dl} \rightarrow \infty$ \cite{jindal2006mimo, yang2012degrees}, where the parameter $\alpha$ is the QSE.\footnote{The Big Theta notation follows the standard Bachmann-Landau order notation, see \cite[footnote 8]{khalilsarai2023fdd}.} The larger the QSE is, the faster the MSE decreases in the high SNR region. The QSEs of the DR, ECSQ, and LJSCC schemes have been derived in \cite{khalilsarai2023fdd}. We further derive the QSE for the proposed TKL scheme.
{\begin{theorem}[QSE for the TKL scheme]\label{theorem:QSE}
    For given training dimension $\beta_{\rm tr}$, feedback dimension $\beta_{\rm fb} \leq \beta_{\rm tr}$, and channels with covaraince rank $r_k = \text{rank}(\Cm_{\hv_k})$, the TKL scheme achieves a channel estimation error of $\Theta({\mathsf{snr}}_{\rm dl}^{-\alpha_{\rm tkl}})$ with probability one over the realizations of the pilot matrix $\Xm$, where the QSE $\alpha_{\rm tkl}$ is given by 
  \begin{align}\label{eq:QSE_TKL}
      \alpha_{\rm tkl} = \mathbb{1}\left\{ \beta_{\rm fb} \geq r_k\right\}.
  \end{align}
\end{theorem}}
\begin{proof}
     See Appendix \ref{sec:QSE_proof}.
\end{proof}

% \begin{remark}
%     \textcolor{blue}{Can we prove the optimality of the TKL scheme for the analog feedback?}
% \end{remark}
\vspace{-5mm}
\section{Multiuser DL Sum Rate}
\vspace{-2mm}
\label{sec:DL_SE}
In this section, we consider the DL achievable ergodic sum-rate based on the estimated DL CSI. 
In the DL data transmission phase, the BS serves $K$ users simultaneously by sending the precoded signal $\xv[n] = \Vm[n] \sv[n]$ at each subcarrier, where $\sv[n] = [s_1[n], \dots, s_K[n]] \sim \mathcal{CN}(\mathbf{0}, \mathbf{I})$ are coded information bearing symbols for users, and where $\Vm[n] \in \mathbb{C}^{M \times K}= [\vv_1[n], \dots, \vv_K[n]]$ is the precoder matrix based on estimated CSI of $K$ users $\widehat{\Hm}[n] = [\widehat{\hv}_1[n], \dots, \widehat{\hv}_K[n]]$ that should satisfy the transmit power constraint\footnote{Note that in practice the transmit power is generally normalized instance by instance, which results in sub-optimal performance. We consider the power constraint to the expectation of the precoding power, which provides the most general result and also benefits the deviation for close-form ergodic rate expression that we will consider later.} 
\begin{align}\label{eq: V_power}
   \tr\left(\mathbb{E}[\Vm^\herm[n] \Vm[n]]\right) \leq \mathsf{snr}_{\rm dl}, \; \forall n \in [N].
\end{align}
The received signal by user $k$ at subcarrier $n$ is given by 
\begin{align}
    y_k[n] &= \hv_k[n]^\herm \Vm[n] \sv[n] + z_{\rm dl}[n] \\
    &= \hv_k[n]^\herm\vv_k[n] s_k[n] + \sum_{j \neq k} \hv_k^\herm[n]\vv_j[n] s_j[n] + z_{\rm dl}[n],
\end{align}
where $z_{\rm dl}[n] \sim \Cc\Nc(0,  1)$ is the DL AWGN.
Denoting the ergodic achievable rate for user $k$ at subcarrier $n$ as $R_k[n]$, we consider the averaged multi-user sum-rate over all subcarriers in [bits/s/Hz], which is given as
\begin{equation}\label{eq:sum_rate}
    R_{\rm avg} \!=\! {\frac{T}{NT}}\sum_{n\in \mathcal{N}_p^c}\sum_{k\in \mathcal{K}}R_k[n] \!+\! \frac{T-T_p}{NT}\sum_{n\in \mathcal{N}_p}\sum_{k\in \mathcal{K}}R_k[n], 
\end{equation}
where $\Nc_p^c$ is the complement of the pilot subcarrier index set, i.e., $\Nc_p^c = \Nc \setminus \Nc_p$. It is noticed that the ergodic rate depends on both the quality of the estimated CSI and the precoding design to mitigate interference among users and keep the desired signal strong. In this work, we consider two commonly used precoders, namely the maximum ratio transmission (MRT) precoder and the zero-forcing (ZF) precoder. Since the precoder of each subcarrier only depends on the CSI of that subcarrier, in the following, we focus on a \textbf{generic subcarrier} and abuse to use the writing $\hv_k, \Cm_{\hv_k}, \Cm_{\widehat{\hv}_k}$ by omitting the subcarrier index for simplicity. Note that the covariance matrices of the CSI and estimated CSI vectors of each subcarrier are the main diagonal $M\times M$ blocks of the corresponding covariance matrices of all subcarriers.

% \textcolor{red}{Comments from Giuseppe: A general question: we have $N$ subcarriers and $N_p$ are probed with pilots. I guess that in the estimation of the channel we are using the full $MN \times MN$ channel covariance and therefore the full MMSE ``interpolator’’ that gives the estimate for all subcarriers. I also assume that the channel estimation error is different for different subcarriers … right? In general, intuitively, a subcarrier ``very far’’ from the pilot subcarriers should suffer from a larger MSE. 
% Is this reflected ??? where do we see this? 
% \textbf{For example, what can we say about the channel coherence bandwidth with respect to how many subcarriers we use?} 
% Are the pilot subcarriers equally spaced? 
% We need to comment on this … basically we have that the dimension in time ($T_p$) of pilots per block depends on 
% the spatial rank, and the dimension in frequency $N_p$ depends on the coherence bandwidth 
% Can we see these effects and make comments maybe in the results section? This would make the results more understandable. We need to say at some point that despite the channel is $N \times M$ dimensional, \textbf{we can probe it with 
% $T_p \times N_p$ symbols and as long as the rank condition is satisfied, at least at the user the MSE estimator has a vanishing MSE as SNR increases.} We need to discuss the forward pilot dimension with respect to the covariance rank, and relate (at least intuitively) the covariance rank with the spatial rank (angular spread etc .. ) and the coherence bandwidth }

Based on \cite[Lemma 1 and 2]{caire2018ergodic}, the upper bound and lower bound (known as the use-and-then-forget (UatF) bound) to the ergodic achievable rate for user $k$ at a generic subcarrier are respectively given by 
\begin{align}
    R_k^{{\rm ub}} &= \bE\left[\log\left(1 + \frac{
    |\hv_k^\herm \vv_k|^2}{ \sum_{j \neq k}  |\hv_k^\herm \vv_j|^2 + 1}\right)\right], \label{eq:R_ub}\\
    R_k^{{\rm UatF}} &= \log\left(1 + \frac{|\mathbb{E}[\hv_k^\herm \vv_k]|^2}{ {\rm Var}(\hv_k^\herm \vv_k) + \sum_{j \neq k} \mathbb{E}[|\hv_k^\herm \vv_j|^2] + 1}\right).\label{eq:R_hard}
\end{align}
It is noticed that in all feedback schemes, for a fixed pilot matrix $\Xm$, the DL CSI for user $k$ can be derived from the channel estimate $\widehat{\hv}_k$ as
\begin{align}\label{eq:h_hat}
    \hv_k = \widehat{\hv}_k + \ev_k, 
\end{align}
where $\widehat{\hv}_k$ and $\ev_k$ are independent with $\widehat{\hv}_k \sim \Cc\Nc(\mathbf{0}, \Cm_{\widehat{\hv}_k})$ and $\ev_k \sim \Cc\Nc(\mathbf{0}, \Cm_{\ev_k} = \Cm_{\hv_k} - \Cm_{\widehat{\hv}_k})$. 
In the following lemmas, we provide closed-form expressions of the UatF  ergodic rate for the considered MRT and ZF precoders, respectively.  

\begin{lemma}[$R_k^{{\rm UatF}}$ under MRT precoder] \label{lemma: MRT}
Using MRT precoding, we have $\Vm = \sqrt{\eta}\widehat{\Hm}$, where the power scaling factor $\eta$ is obtained based on \eqref{eq: V_power} as 
\begin{align}
    \eta = \frac{\mathsf{snr}_{\rm dl}}{\tr\left(\mathbb{E}[\widehat{\Hm}^\herm \widehat{\Hm}]\right)} 
    = \frac{\mathsf{snr}_{\rm dl}}{\tr\left(\sum_{k=1}^K \Cm_{\widehat{\hv}_k}\right)}. \label{eq:eta_MRT}
\end{align}
The UatF ergodic rate under MRT precoder in closed-form is 
\begin{align}
    R^{{\rm MRT}}_{k} &= \log\left(1 +\frac{\eta\trace^2(\Cm_{\widehat{\hv}_k})}{\eta\tr(\Cm_{\widehat{\hv}_k} \Cm_{{\hv}_k})+ \eta \sum_{j \neq k} \trace( \Cm_{\widehat{\hv}_j} \Cm_{\hv_k})+ 1}\right), \\
    &=\log\left(1 +\frac{\trace^2(\Cm_{\widehat{\hv}_k})}{  \tr\left(\left(\sum_{k=1}^K \Cm_{\widehat{\hv}_k}\right) \Cm_{\hv_k}\right)+ \frac{1}{\eta}}\right).\label{eq:R_hard_MRT}
\end{align}
\end{lemma}
\begin{proof}
    See Appendix \ref{sec:MRT_lemma}.
\end{proof}

\begin{lemma}[$R_k^{{\rm UatF}}$ under ZF precoder]\label{lemma: ZF}
    Using ZF precoding, the precoder matrix is a power normalized pseudo-inverse of the estimated channel matrix, i.e., $\Vm =  \sqrt{\widetilde{\eta}} \widehat{\Hm} (\widehat{\Hm}^{\herm} \widehat{\Hm})^{-1}$, where the power scaling factor $\widetilde{\eta}$ is obtained based on \eqref{eq: V_power} as 
\begin{align}
    \widetilde{\eta} &= \frac{\mathsf{snr}_{\rm dl}}{\tr\left(\mathbb{E}\left[\left(\widehat{\Hm}^\herm \widehat{\Hm}\right)^{-1}\right]\right)}. \label{eq:V_eta_ZF}
\end{align}
The UatF ergodic rate under ZF precoder can be written as  
\begin{align}
    {R}^{{\rm ZF}}_{k} &= \log\left(1 + \frac{1}{  \tr\left(\mathbb{E}\left[\widehat{\Hm} (\widehat{\Hm}^{\herm} \widehat{\Hm})^{-2} \widehat{\Hm}^{\herm}\right]\Cm_{\ev_k}\right) + \frac{1}{\widetilde{\eta}}}\right).  \label{eq:R_hard_ZF} 
\end{align} 
% \begin{align}
%     {R}^{{\rm ZF}}_{k} &= \log\left(1 + \frac{1}{  \trace\left(\mathbb{E}\left[\widehat{\Hm} (\widehat{\Hm}^{\herm} \widehat{\Hm})^{-2} \widehat{\Hm}^{\herm}\right]\Cm_{\ev_k}\right) + \frac{1}{\widetilde{\eta}}}\right)  \label{eq:R_hard_ZF} \\
%     &\geq  \log\left(1 + \frac{1}{ \|\Cm_{\ev_k}\|_2 \trace\left(\mathbb{E}\left[(\widehat{\Hm}^{\herm} \widehat{\Hm})^{-1}\right]\right) + \frac{1}{\widetilde{\eta}}}\right) \triangleq \check{R}_k^{\rm ZF} \label{eq:R_ZF_low}.
% \end{align} 
\end{lemma}
\begin{proof}
    See Appendix \ref{sec:proof_ZF}.
\end{proof}
\begin{remark}
    Note that the expression in \eqref{eq:R_hard_ZF} is not in closed-form, since the terms $\mathbb{E}[(\widehat{\Hm}^\herm \widehat{\Hm})^{-1}]$ in \eqref{eq:V_eta_ZF} and $\mathbb{E}[\widehat{\Hm} (\widehat{\Hm}^{\herm} \widehat{\Hm})^{-2} \widehat{\Hm}^{\herm}]$ in \eqref{eq:R_hard_ZF} are required to be numerically calculated via Monte Carlo method. 
    It is still an open problem to obtain the closed-form expressions for both terms based on the corresponding covariance knowledge, i.e., $\{\Cm_{\widehat{\hv}_k}\}_{k=1}^K$.  \hfill $\lozenge$
\end{remark}

% Given the ergodic rate $R_k^{\ell}$ of user $k$, the sum rate is defined as
% $R^{\ell} = \sum_{k=1}^K R_k^{\ell}$, for $\ell \in \{\text{ub}, \text{UatF}, \text{MRT}, \text{ZF}\}$.

\section{Simulation Results}
\label{sec:simulation}
In our simulation, we consider that the BS is equipped with a uniform linear array (ULA) of $M=32$ antennas, simultaneously serving $K=6$ users over $N=32$ OFDM subcarriers. The resulting wideband channel dimension is $MN = 1024$. We consider an OFDM subcarrier spacing of $\Delta_f = 30~$kHz, and the resulting symbol duration equals to $T_{\rm u} = 1/\Delta_{f} \approx 33 ~\mu$s. Assuming that the maximum channel delay spread is equal to $\tau_{\rm max} = 7~\mu$s, and that the duration of the OFDM cyclic prefix is equal to the maximum channel delay spread, it results in an OFDM symbol duration of $T_{\rm s}= \tau_{\max} + T_{\rm u} = 40 ~\mu$s. Then, the coherence bandwidth is $B_{\rm c} = 1/T_{\rm s} \approx 143 ~$kHz, and the total bandwidth is $B = N \Delta f \approx 1~$MHz. The channel is assumed to be constant over a coherence time of $T_{\rm c}= 1~$ms, corresponding to $T= T_{\rm c}/T_{\rm s} = 25 ~$OFDM symbols in time. 
Hence the total time-frequency resource dimension is given by $TN = 25 \times 32 = 800$, where $\beta_{\rm tr}$ dimensions are dedicated to DL training, while the rest dimensions are used for data transmission. Note that the UL and DL SNR have the relation of $\mathsf{snr}_{\rm dl} = K\mathsf{snr}_{\rm ul}$ as explained in the system model.

For the channel modeling, we consider the multipath component channel with $L$ number of multipath. The instantaneous DL channel vector for the $k$-th user is given by 
\begin{align}\label{eq:channel_matrix}
    \hv_k = {\rm vec}\left(\sum_{\ell=1}^L  g_{k, \ell} \av(\theta_{k, \ell}) \bv^{\transp}(\tau_{k, \ell}) \right),
\end{align}
where $g_{k, \ell}  \sim \Cc\Nc(0, \gamma_{k,\ell})$ is the complex gain of the $\ell$-th path of user $k$ with its power $\gamma_{k, \ell}$, and where $\theta_{k, \ell}$ and $\tau_{k, \ell}$ are respectively the angle-of-departure (AoD) and the delay of the $\ell$-th path of user $k$, and where $\av(\theta_{k, \ell}) \in \mathbb{C}^M$ is the steering vector of the ULA, whose $m$-th element is $[\av(\theta_{k, \ell})]_{m} = e^{j  \pi  (m-1) \sin(\theta_{k, \ell})}, m \in [M]$ by assuming the antenna spacing equaling to half of the carrier wavelength, and $\bv(\tau_{k, \ell}) \in \mathbb{C}^{N}$ contains the phase rotations corresponding to the path delay $\tau_{k, \ell}$, whose $n$-th element is $[\bv(\tau_{k, \ell})]_n = e^{-j 2 \pi (n-1) {\Delta f} \tau_{k, \ell}},  n  \in [N]$.  
Based on the channel vector in \eqref{eq:channel_matrix}, the DL covariance matrix of user $k$ is given as
% \footnote{\text{where $\otimes$ denotes the kronecker product operation.}}
\begin{align}\label{eq: C_h}
    \Cm_{\hv_k} \!=\! \sum_{\ell=1}^L \gamma_{k, \ell} \left(\bv(\tau_{k, \ell}) \bv(\tau_{k, \ell})^\herm\right) \otimes\left(\av(\theta_{k, \ell}) \av(\theta_{k, \ell})^\herm\right). 
\end{align}

In all simulations, the AoDs, delay, and the path power are generated independently from uniform distributions, i.e., $\theta_{k, \ell} \sim {\rm Unif}[-60^{\circ}, 60^{\circ}]$, $\tau_{k, \ell} \sim {\rm Unif}[0, \tau_{\rm max}]$ and $\gamma_{k, \ell} \sim {\rm Unif}[0.4, 0.8]$ where the total path power is normalized as $\sum_{\ell=1}^L \gamma_{k, \ell} = 1$. 
We provide the simulation results for two scenarios: sparse scattering with $L=6$ paths per user, and rich scattering with $L=60$ paths per user.
Note that the rank of the channel covariance matrix is equal to the number of paths, i.e., $r_k = L$, and we focus on the regime where the training dimension is sufficient ($\beta_{\rm tr} \geq L$) while the feedback dimension is a bottleneck. In our simulation, we consider $\beta_{\rm tr}= 64$ training dimension with $N_p = 8$ selected subcarriers and $T_p= 8$ pilot sequences.
We consider the normalized MSE (NMSE) defined as $\frac{1}{K} \sum_{k=1}^K \frac{\mathbb{E}[\|\widehat{\hv}_k -\hv_k\|^2]}{\mathbb{E}[\|\hv_k\|^2]}$ as the channel estimation metric.
The numerical results are averaged over $10$ channel geometries and each with $1000$ channel realizations. 
% \textcolor{blue}{ add here how to choose subcarrier to probe pilot signals based on coherence bandwidth, and plot the channel estimation MSE vs. the number of pilot dimension of single subcarrier.}

\begin{remark}[Pilot arrangement in frequency]
    {We know that as long as the pilot dimension $\beta_{\rm tr}$ is larger than the channel rank, the MSE at the user side vanishes as the DL SNR increases. However, given a fixed $\beta_{\rm tr}$, one still needs to decide  $T_{\rm p}$ and $ N_{\rm p}$ with the constraint $\beta_{\rm tr} = T_{\rm p} N_{\rm p}$. We notice that the channel remains relatively flat over about $N_{\rm cb} \approx \frac{B_{\rm c}}{\Delta_f} \approx 4$ consecutive subcarriers. Thus, the entire bandwidth can be divided into $N_{\rm p} = \frac{N}{N_{\rm cb}} = 8$ coherence band, where for each coherence band we probe its central subcarrier and in total we have a comb-type pilot arrangement with $N_{\rm p} = T_{\rm p} = 8$. In Fig. \ref{fig:select_sub}, we compare the comb-type pilot arrangement to the single pilot case where only subcarrier-16 is probed. Note that in both cases we have $\beta_{\rm tr} = T_{\rm p} N_{\rm p} = 64$. It is obvious when only a single subcarrier is probed, the subcarriers ``very far'' from the pilot subcarrier suffer from a larger NMSE, while the comb-type pilots result in a much lower NMSE over all subcarriers.  \hfill $\lozenge$
     }
     % We focus on the selection of subcarriers $\mathcal{N}_{\rm p}$ and pilot sequences $T_{\rm p}$ for DL probing. As shown in Fig. \ref{fig:select_sub}, when only the central subcarrier (index 16) is probed, the subcarriers ``very far'' from the pilot subcarrier suffer from a larger MSE. Moreover, the channel remains relatively flat over about $N_{\rm cb} \approx \frac{B_{\rm c}}{\Delta_f} \approx 4$ consecutive subcarriers. Thus, the entire bandwidth can be divided into $N_{\rm p} = \frac{N}{N_{\rm cb}} = 8$ blocks. By selecting pilot subcarriers in a comb-type arrangement, we can probe DL signals with comparable channel estimation error while minimizing the use of subcarriers, as confirmed in Fig. \ref{fig:select_sub}. Finally, we choose pilot sequences $T_{\rm p}$ such that the pilot dimension $T_{\rm p} \times N_{\rm p}$ is larger than the channel rank, ensuring that the user achieves a vanishing MSE as the DL SNR increases.}  
\end{remark}

\begin{figure}
        \centering
        \includegraphics[scale=0.4]{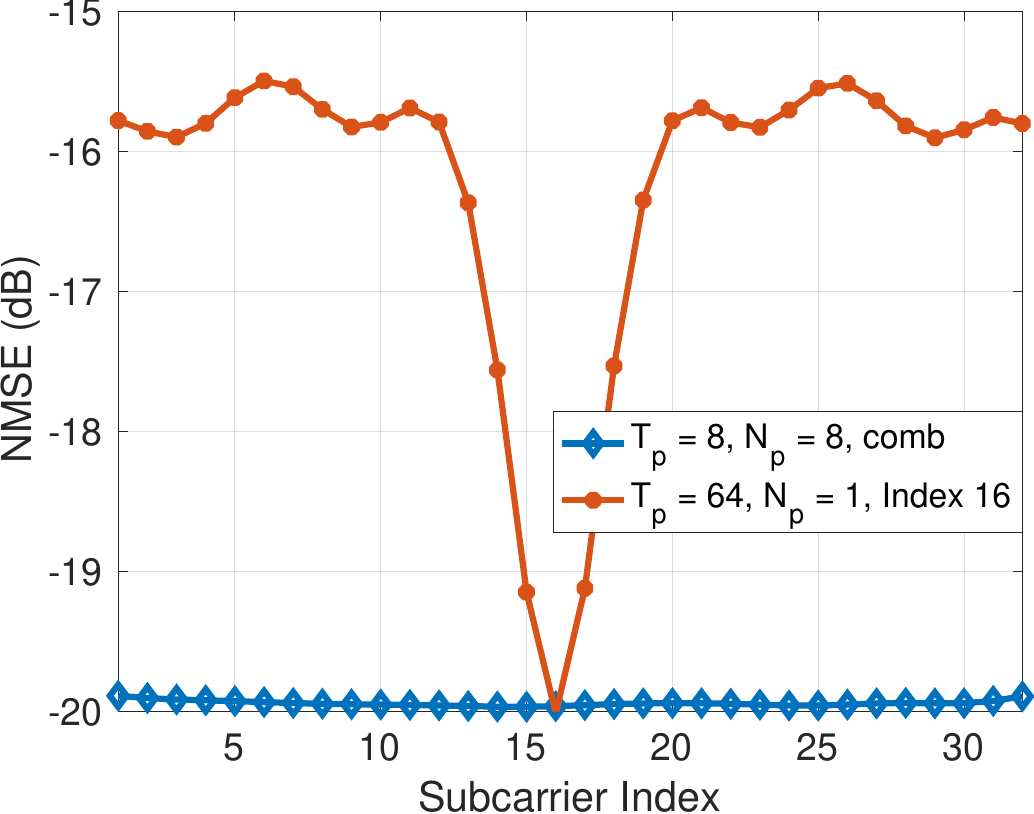}
        % \vspace{-2mm}
        \caption{The average NMSE at the user over $100$ geometries of different subcarriers under single probed subcarrier (index 16) and under 8 probed subcarriers with comb-type pilot arrangement when SNR $= 10$ dB and $L = 6$.}
        \label{fig:select_sub}
        \vspace{-6mm}
    \end{figure}

% \subsection{$R^{\rm UtaF}$ versus $R^{\rm UtaF2}$ under MRT and ZF precoding}
% First, we need to check the derivation of Lemma \ref{lemma: MRT} and Lemma \ref{lemma: ZF} in the simulations. As shown in Fig. \ref{}, we plot the use-and-then-forget bound vs. DL SNR 
\begin{figure*}[ht!]
\centering
        \begin{subfigure}[b]{0.4\textwidth}
        \includegraphics[width=1\columnwidth]{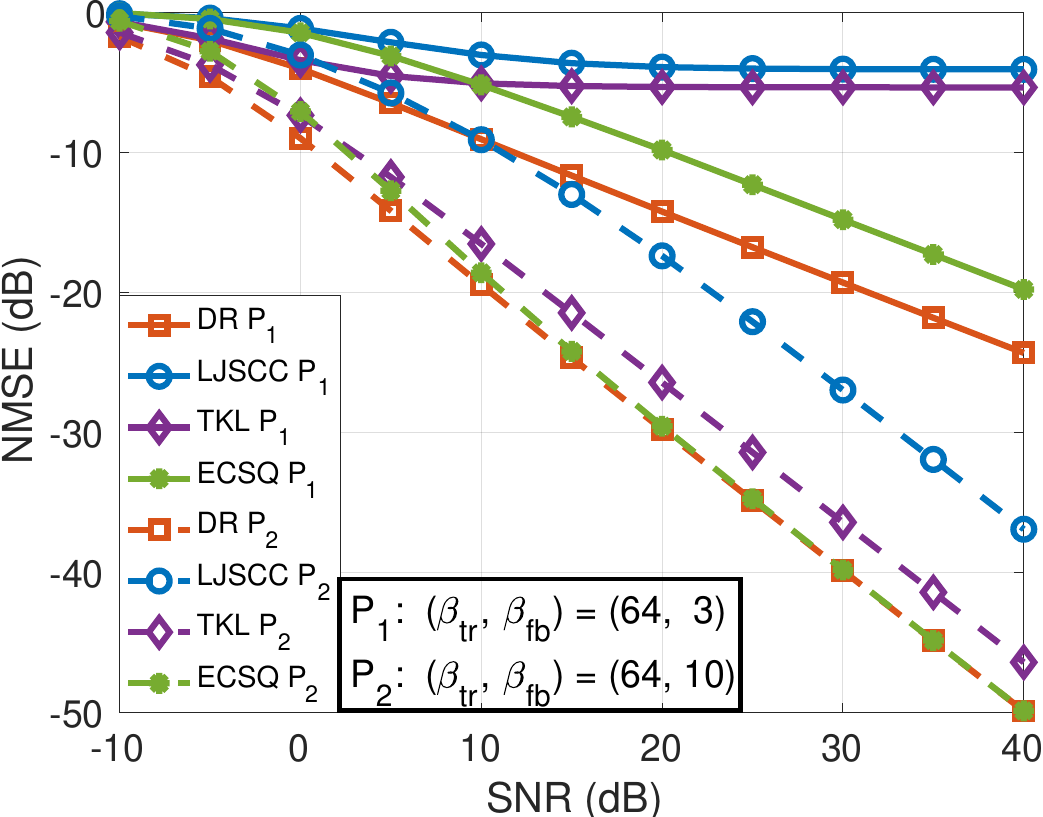}
        \caption{$L = 6$}
        \label{fig:L_6}
        \end{subfigure}
        ~
        \begin{subfigure}[b]{0.4\textwidth}
        \includegraphics[width=\columnwidth]{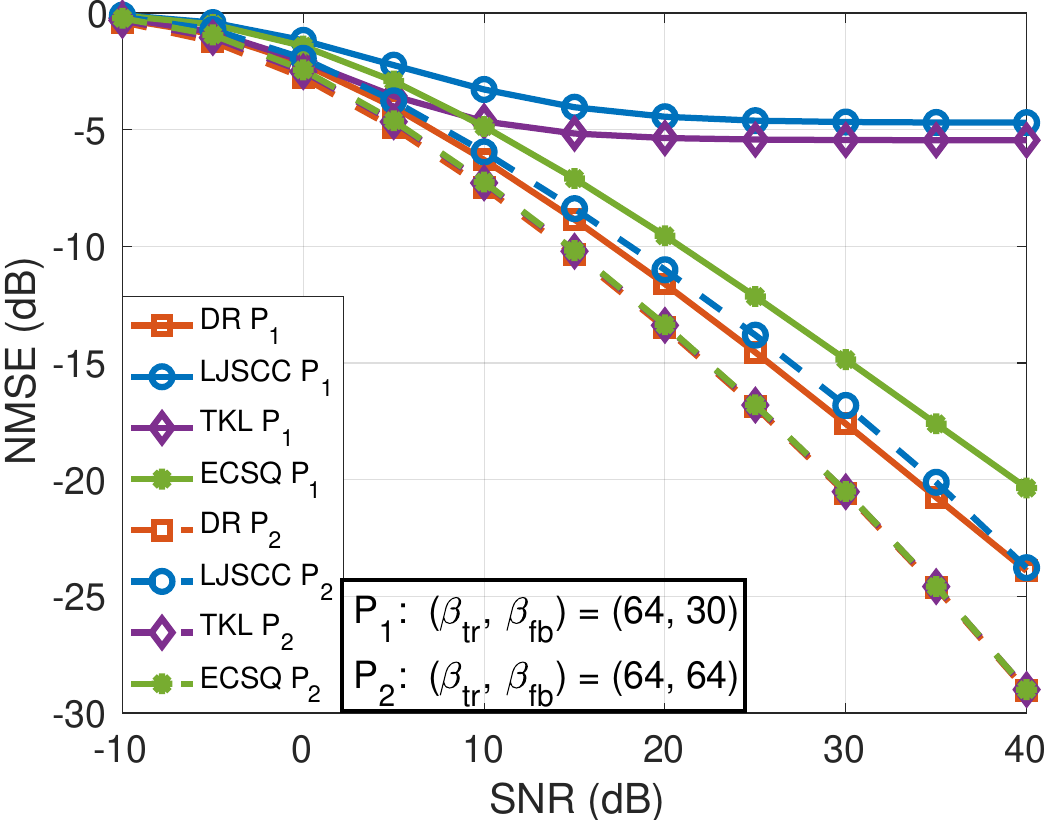}
        \caption{$L = 60$}
        \label{fig:L_60}
        \end{subfigure}
 \caption{NMSE (in dB) vs. DL SNR (in dB) under sparse (a) and rich (b) scattering channels. }
     \label{fig:NMSE}
     \vspace{-6mm}
\end{figure*}

\subsection{Normalized MSE vs. DL SNR for Channel Estimation}
The results under sparse and rich scattering channels are shown in Fig. \ref{fig:L_6} and \ref{fig:L_60}, respectively, where $P_1$ (in solid lines) and $P_2$ (in dashed lines) denote the cases with compressed ($\beta_{\rm fb} < L$) and sufficient ($\beta_{\rm fb} > L$) feedback dimensions, respectively. Note that the slopes in the high SNR region represent the QSEs of different schemes. From the results in both Fig. \ref{fig:L_6} and \ref{fig:L_60}, we observe that in the high SNR region ($20$dB to $40$dB), the slopes of the TKL scheme are almost zero in $P_1$ and almost one in $P_2$, which confirms the analytical result of the QSE of the TKL scheme in {{Theorem}} \ref{theorem:QSE}. Moreover, in $P_1$ with insufficient feedback, as SNR increases, the NMSE of the ECSQ scheme decreases while the NMSE of the TKL scheme becomes flat due to its zero QSE. In $P_2$ with sufficient feedback, the QSEs of all schemes are the same, and the proposed TKL scheme performs much better than the LJSCC scheme and even very close to the ECSQ scheme especially under rich scattering channels $L=60$.
Furthermore, in all cases, the proposed TKL scheme outperforms the LJSCC and ECSQ schemes under small SNR ($-10$ dB to $5$ dB).

\subsection{DL Sum Rate vs. SNR}
In this part, we evaluate the DL sum-rate of different feedback schemes under various SNR. The results of MRT and ZF precoding are shown in Fig. \ref{fig:DL_MRT_snr} and Fig. \ref{fig:DL_ZF_snr}, respectively.

%%%%%%%%%%%%%%%%%%%%
\subsubsection{MRT precoding}
In the case of compressed feedback dimensions, as shown in Fig. \ref{fig:L_6_beta_fb_3_MRT} and Fig. \ref{fig:L_60_beta_fb_30_MRT}, when the DL SNR ranges from $-10$ dB to $10$ dB (a practical DL SNR range), the TKL scheme always outperforms the ECSQ scheme and show quite close performance to the DR scheme. Note that the TKL scheme is even aligned with the DR scheme when $L=60$ in Fig. \ref{fig:L_60_beta_fb_30_MRT} when the SNR is low ($-10$ dB to $5$ dB), indicating the advantages of the TKL scheme. As the SNR increases, since the QSE of the ECSQ scheme is larger than that of the TKL scheme, the sum rate of the ECSQ scheme increases faster and reaches the DR scheme in the high SNR region ($30$ dB to $40$ dB), while the TKL scheme becomes flat in the high SNR region. 
Also note that when the SNR ranges from $-10$ dB to $0$ dB, the sum rate for the ECSQ is almost zero, since the rate constraint ($\beta_{\rm fb} C_{\rm ul}$) is less than the penalty of the replaced scalar quantization, i.e., $1. 508 \sum_{i=1}^{MN} \mathbb{1}\left\{\lambda_i^{\uv_k}> 0\right\}$, resulting in no information transfer during feedback. Therefore, for the small UL rate (small $\beta_{\rm fb}$ or SNR), the ECSQ scheme encounters a bottleneck due to the replaced scalar quantization compared to the DR scheme.
In the case of sufficient feedback dimension, as shown in Fig. \ref{fig:L_6_beta_fb_10_MRT} and Fig. \ref{fig:L_60_beta_fb_64_MRT}, the four schemes perform in a similar trend and match each other. The TKL scheme performs slightly better than the ECSQ scheme especially when the SNR is low ($-10$dB to $5$dB).   

\subsubsection{ZF precoding}
In the case of compressed feedback dimension, when the SNR is in the range of $-10$ dB to $10$ dB, the TKL scheme performs better than the other practical schemes and gets close to the DR scheme, even matching the DR scheme when $L=60$ in Fig. \ref{fig:L_60_beta_fb_30_ZF}. As the SNR increases, since ZF precoding is preferred over MRT precoding in the high SNR region in practice, and the MSEs of the channel estimation for the DR and ECSQ schemes also decrease faster, the sum-rates for the two schemes increase. On the contrary, since the QSEs of the TKL and LJSCC schemes are all zero, the channel estimation MSE becomes constant. Therefore, the channel estimation error dominates in the case of compressed feedback dimensions, i.e., no matter how large the DL SNR becomes, the sum-rates for the TKL and LJSCC schemes become flat. Moreover, in the case of sufficient feedback dimension, as illustrated in Fig. \ref{fig:L_6_beta_fb_10_ZF} and Fig. \ref{fig:L_60_beta_fb_64_ZF}, the four schemes except the LJSCC scheme coincide with each other when the DL SNR ranges from $-10$ dB to $40$ dB. Taking a closer look, when the SNR is $0$ dB, the TKL scheme performs slightly better than the ECSQ scheme.

\begin{figure*}[ht!]
\centering
        \begin{subfigure}[b]{0.4\textwidth}
        \includegraphics[width=1\columnwidth]{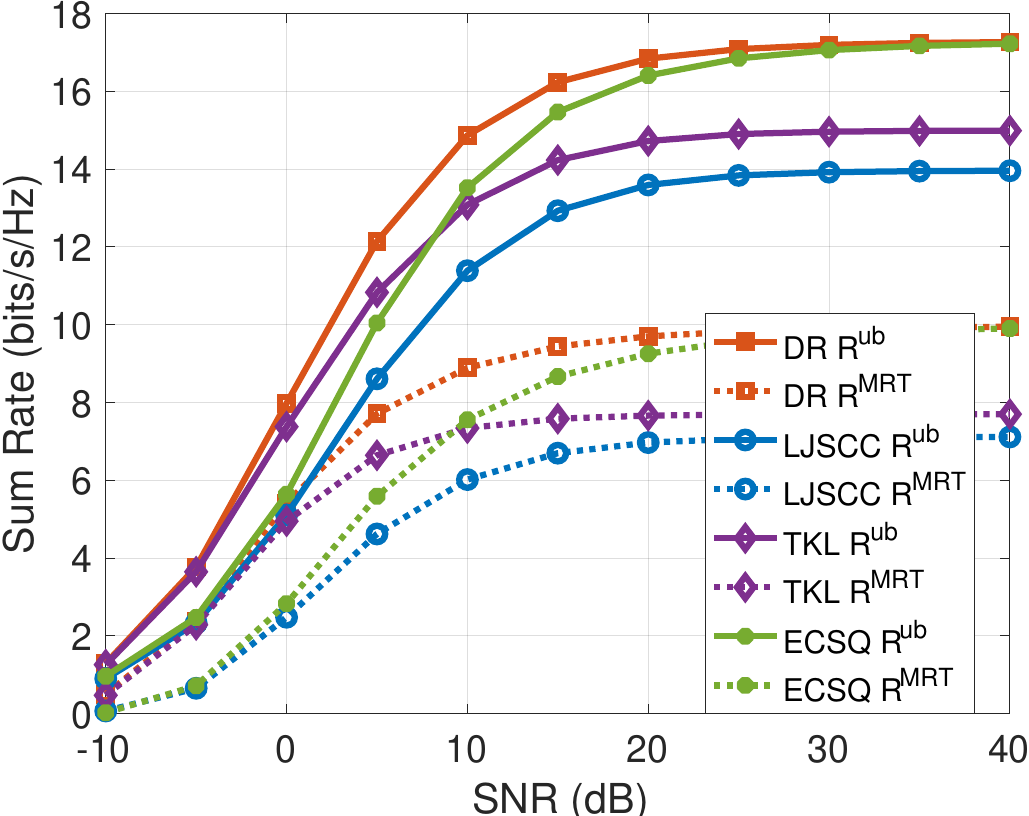}
        \caption{$L = 6$ and $\beta_{\rm fb} = 3$.}
        \label{fig:L_6_beta_fb_3_MRT}
        \end{subfigure}
        ~
        \begin{subfigure}[b]{0.4\textwidth}
        \includegraphics[width=\columnwidth]{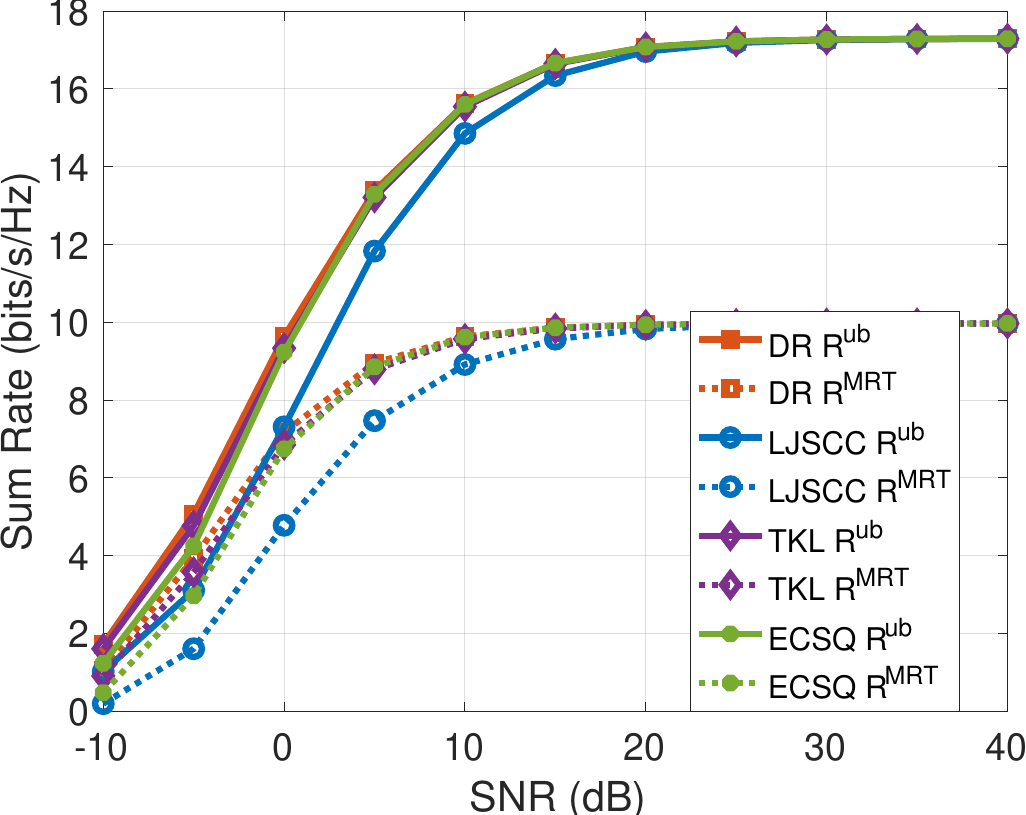}
        \caption{$L = 6$ and $\beta_{\rm fb} = 10$.}
        \label{fig:L_6_beta_fb_10_MRT}
        \end{subfigure}
        ~
        \begin{subfigure}[b]{0.4\textwidth}
        \includegraphics[width=1\columnwidth]{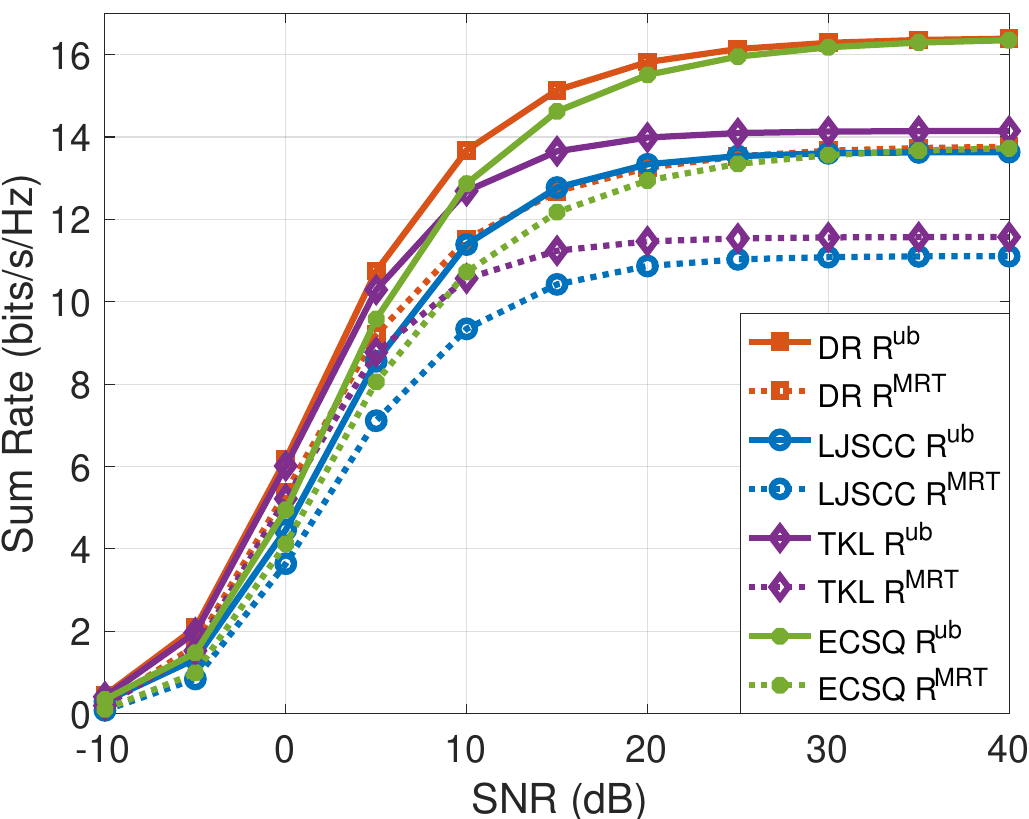}
        \caption{$L = 60$ and $\beta_{\rm fb} = 30$.}
        \label{fig:L_60_beta_fb_30_MRT}
        \end{subfigure}
        ~
        \begin{subfigure}[b]{0.4\textwidth}
        \includegraphics[width=\columnwidth]{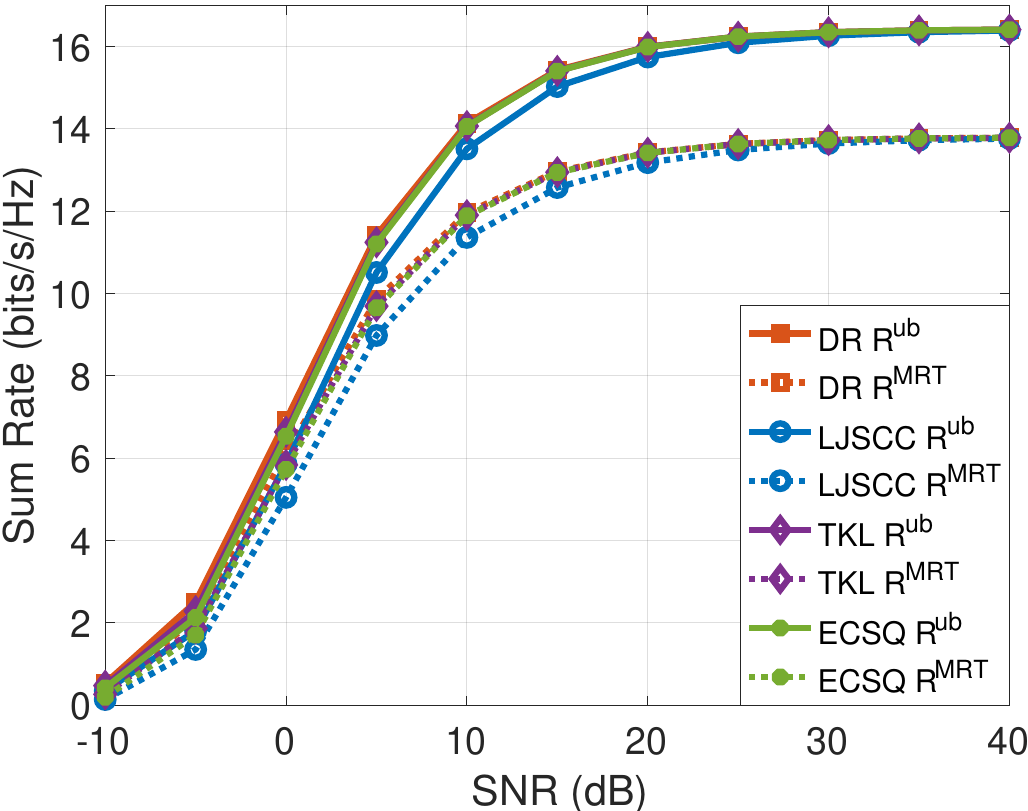}
        \caption{$L = 60$ and $\beta_{\rm fb} = 64$.}
        \label{fig:L_60_beta_fb_64_MRT}
        \end{subfigure}
\caption{DL sum-rate vs. DL SNR under MRT precoding.}
     \label{fig:DL_MRT_snr}
     \vspace{-6mm}
\end{figure*}

\begin{figure*}[ht!]
\centering
        \begin{subfigure}[b]{0.4\textwidth}
        \includegraphics[width=\columnwidth]{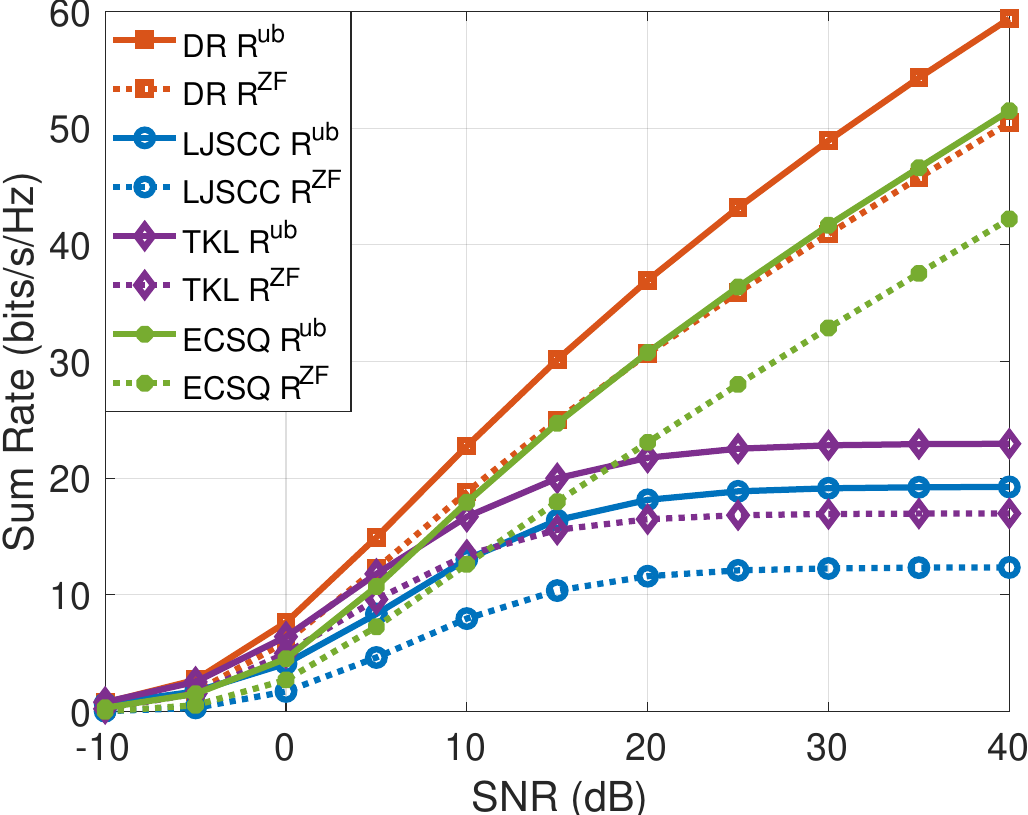}
        \caption{$L = 6$ and $\beta_{\rm fb} = 3$.}
        \label{fig:L_6_beta_fb_3_ZF}
        \end{subfigure}
        ~
        \begin{subfigure}[b]{0.4\textwidth}
        \includegraphics[width=\columnwidth]{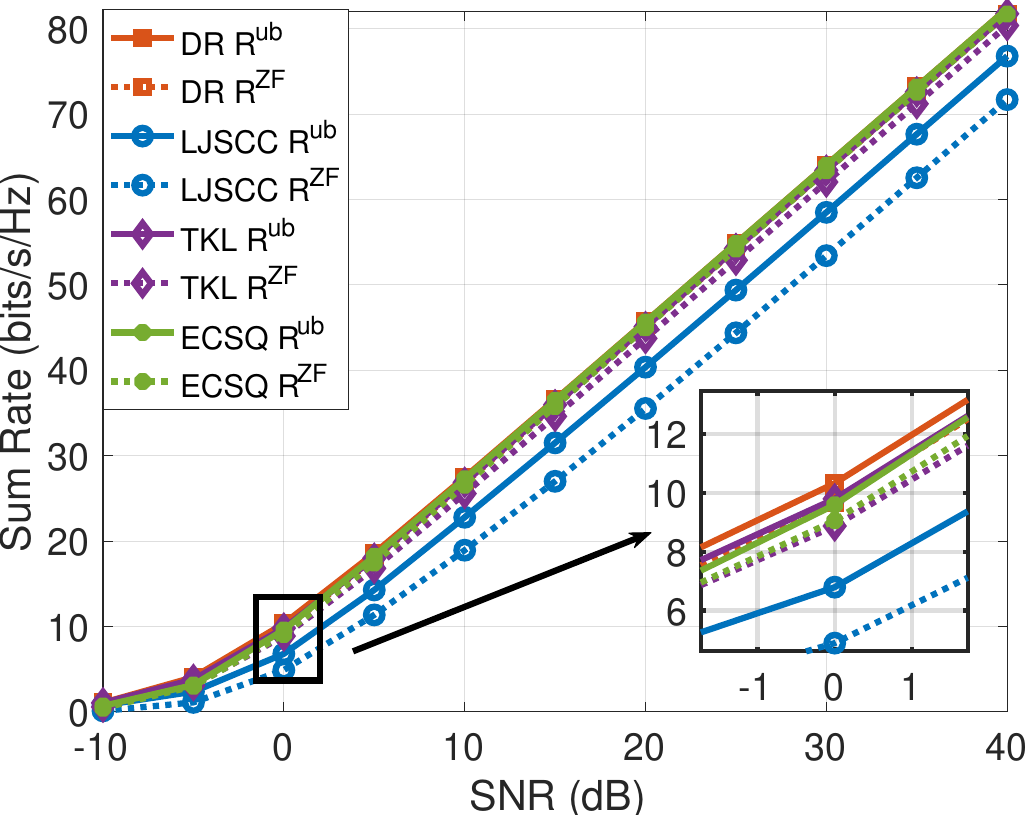}
        \caption{$L = 6$ and $\beta_{\rm fb} = 10$.}
        \label{fig:L_6_beta_fb_10_ZF}
        \end{subfigure}
        ~
        \begin{subfigure}[b]{0.4\textwidth}
        \includegraphics[width=1\columnwidth]{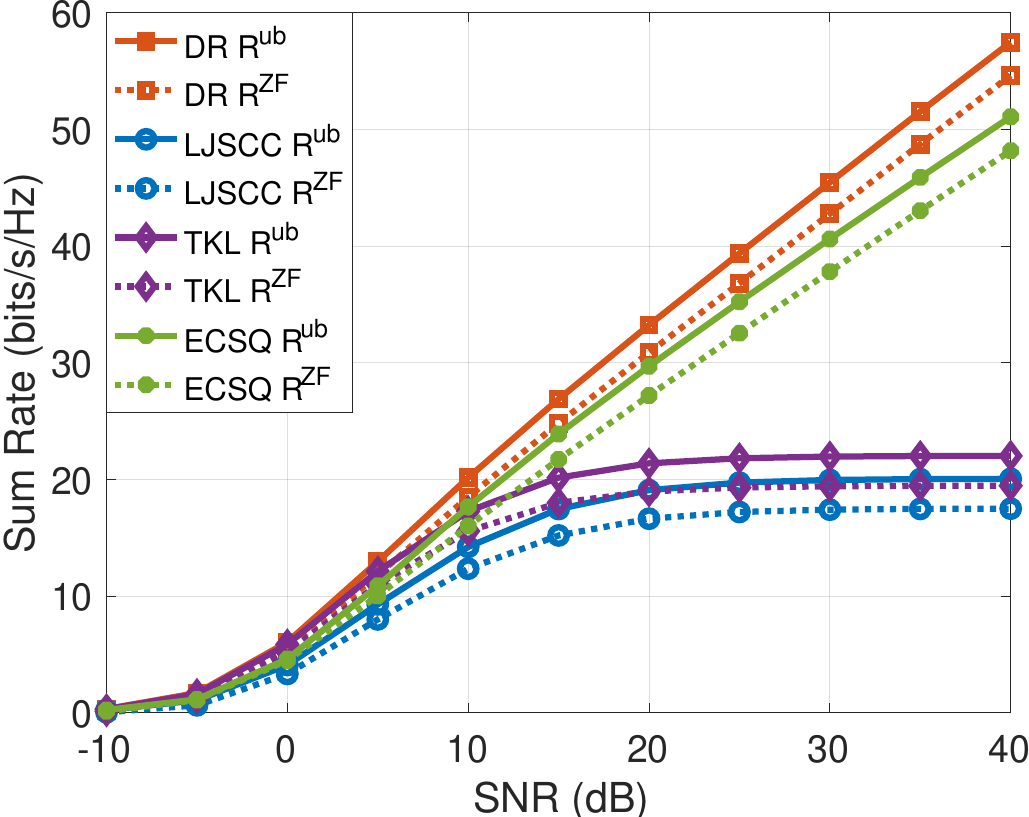}
        \caption{$L = 60$ and $\beta_{\rm fb} = 30$.}
        \label{fig:L_60_beta_fb_30_ZF}
        \end{subfigure}
        ~
        \begin{subfigure}[b]{0.4\textwidth}
        \includegraphics[width=\columnwidth]{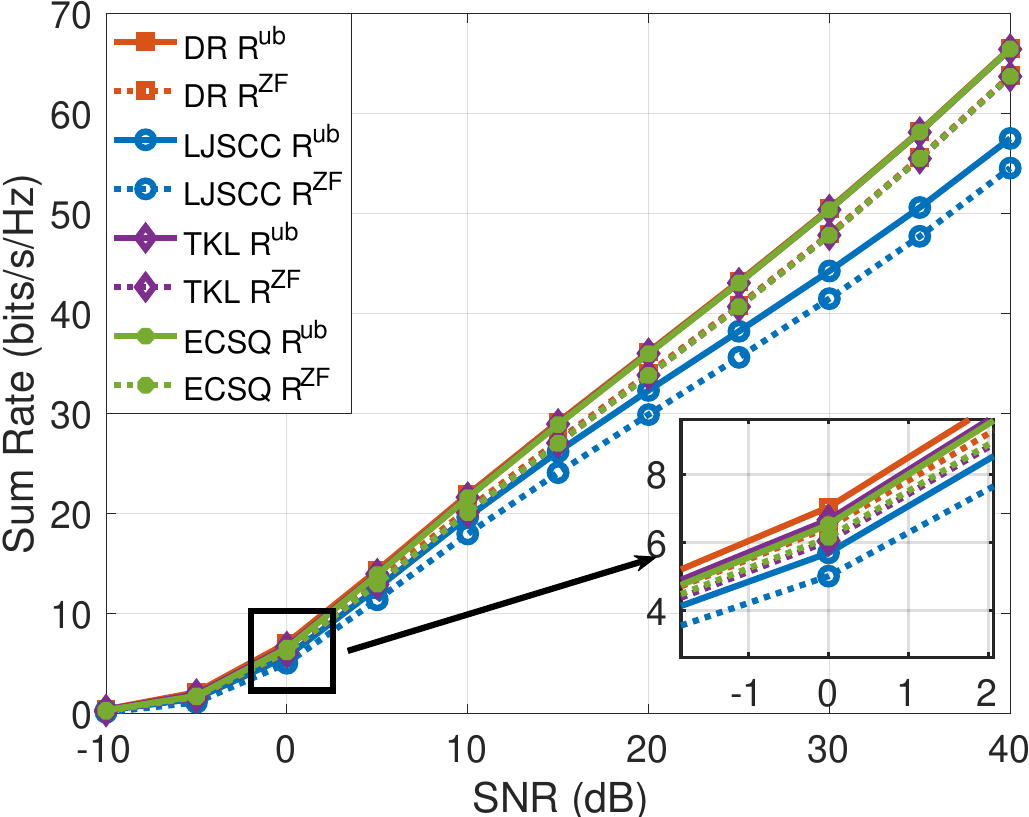}
        \caption{$L = 60$ and $\beta_{\rm fb} = 64$.}
        \label{fig:L_60_beta_fb_64_ZF}
        \end{subfigure}
 \caption{DL sum-rate vs. DL SNR under ZF precoding.}
     \label{fig:DL_ZF_snr}
     % \vspace{-3mm}
\end{figure*}

\subsection{DL Sum Rate vs. Feedback Dimension $\beta_{\rm fb}$}
Next, we evaluate the effect of the feedback dimension $\beta_{\rm fb}$ under a fixed DL SNR. Specifically, 
the results of MRT precoding and ZF precoding are shown in Fig. \ref{fig:DL_ZF_beta_fb}, where we plot how the two DL sum-rates under MRT (ZF) precoding ($R^{\rm ub}$ and $R^{\rm MRT}$ ($R^{\rm ZF}$)) vary as $\beta_{\rm fb}$ changes for the sparse channels with $L=6$ and the rich-scattering channels with $L=60$ when ${\rm SNR}= \{5, 10\}$ dB. Since $R^{\rm ub}$ and $R^{\rm UatF}$ share the same trend, we do not present $R^{\rm ub}$ for a clear view.
From Fig. \ref{fig:DL_ZF_beta_fb} we observe that in both sparse and rich-scattering channels, it is evident that the TKL scheme outperforms the other practical schemes as $\beta_{\rm fb}$ changes, especially under the compressed feedback dimensions. The gap between the DR scheme and the TKL schemes becomes smaller as $\beta_{\rm fb}$ increases. Also, as $\beta_{\rm fb}$ increases and $\beta_{\rm fb} \geq L$, the ECSQ scheme gradually approaches the TKL scheme.

\begin{figure*}[ht!]
\centering
        \begin{subfigure}[b]{0.23\textwidth}
        \includegraphics[width=1\columnwidth]{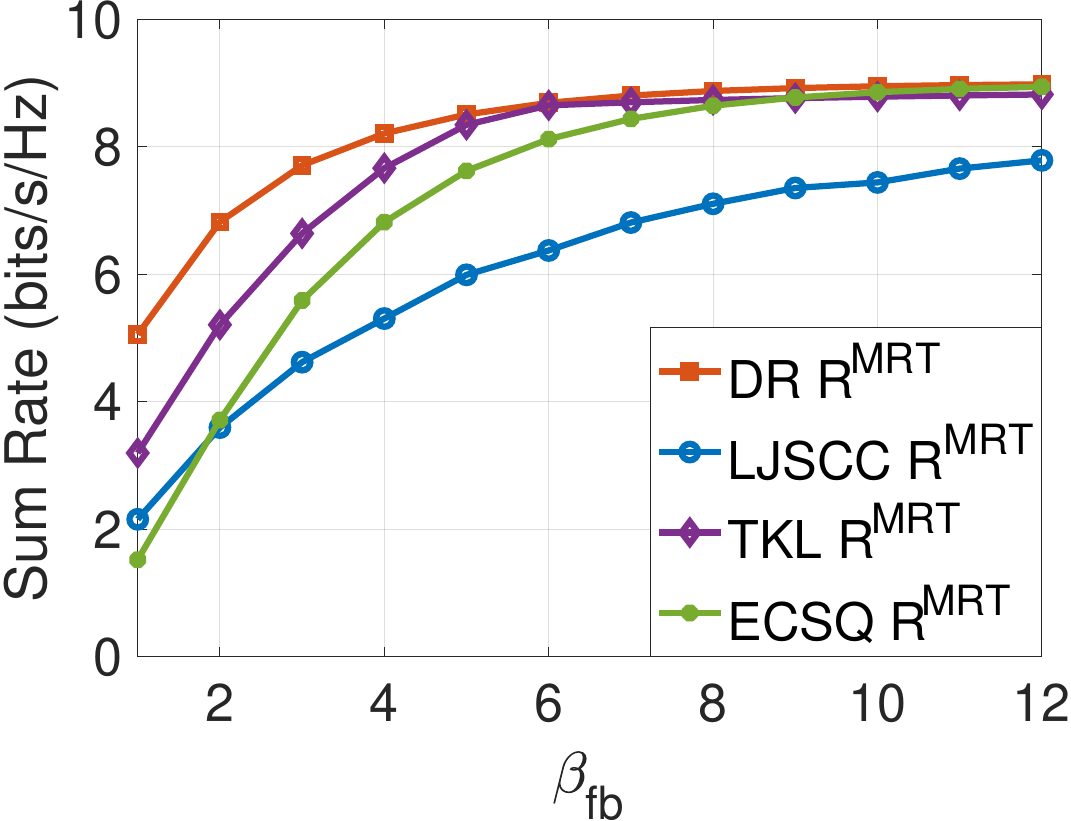}
        \caption{$L = 6$ and ${\rm SNR}= 5~$dB.}
        \label{fig:L_6_snr_5_MRT}
        \end{subfigure}
        ~
        \begin{subfigure}[b]{0.23\textwidth}
        \includegraphics[width=\columnwidth]{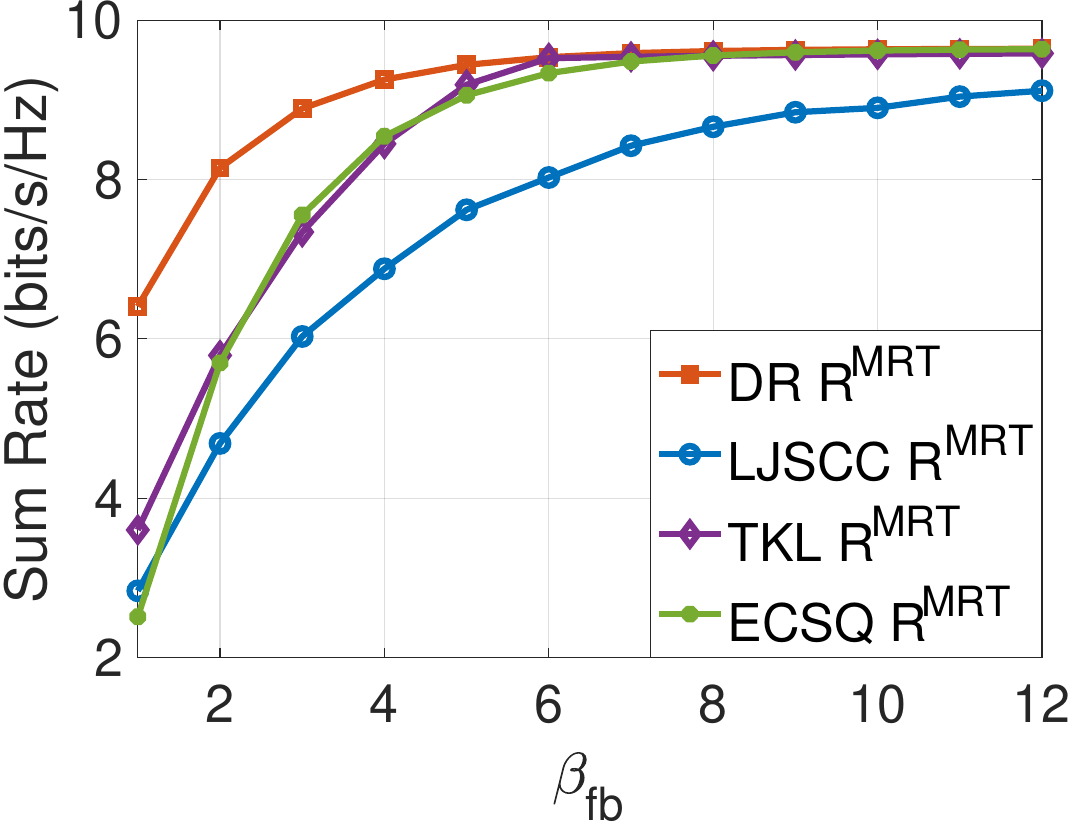}
        \caption{$L = 6$ and ${\rm SNR}= 10~$dB.}
        \label{fig:L_6_snr_10_MRT}
        \end{subfigure}
        ~
        \begin{subfigure}[b]{0.23\textwidth}
        \includegraphics[width=1\columnwidth]{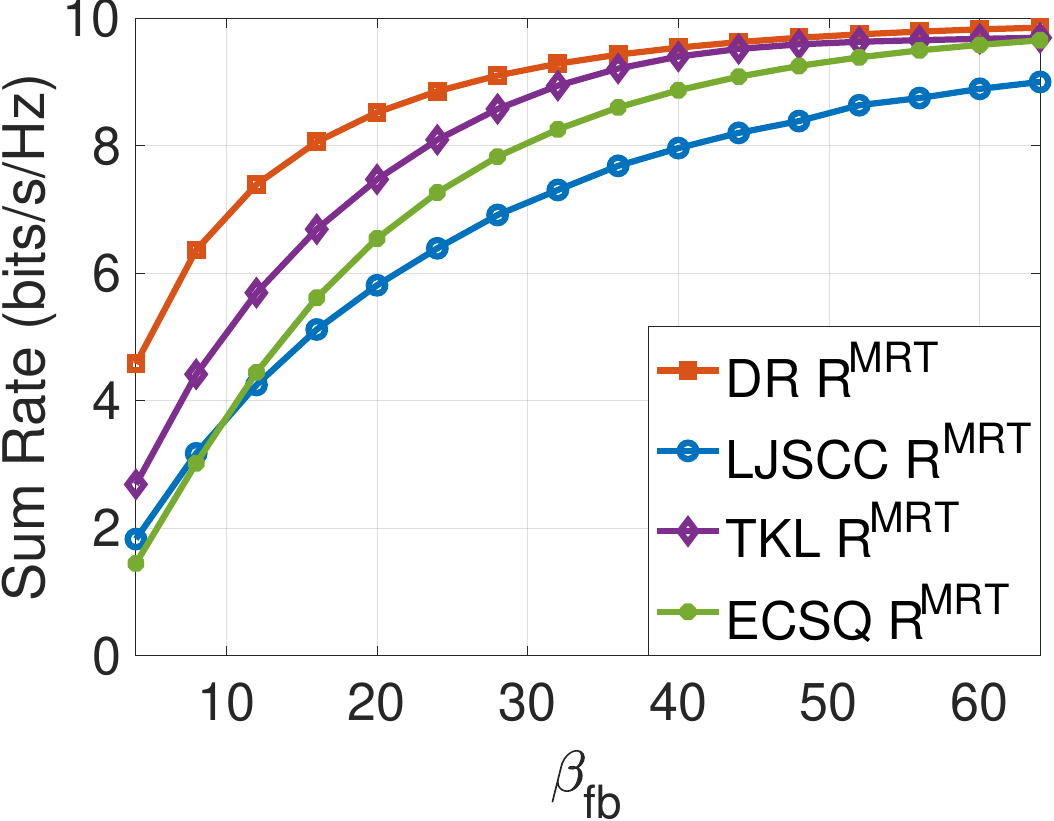}
        \caption{$L = 60$ and ${\rm SNR}= 5~$dB.}
        \label{fig:L_60_snr_5_MRT}
        \end{subfigure}
        ~
        \begin{subfigure}[b]{0.23\textwidth}
        \includegraphics[width=\columnwidth]{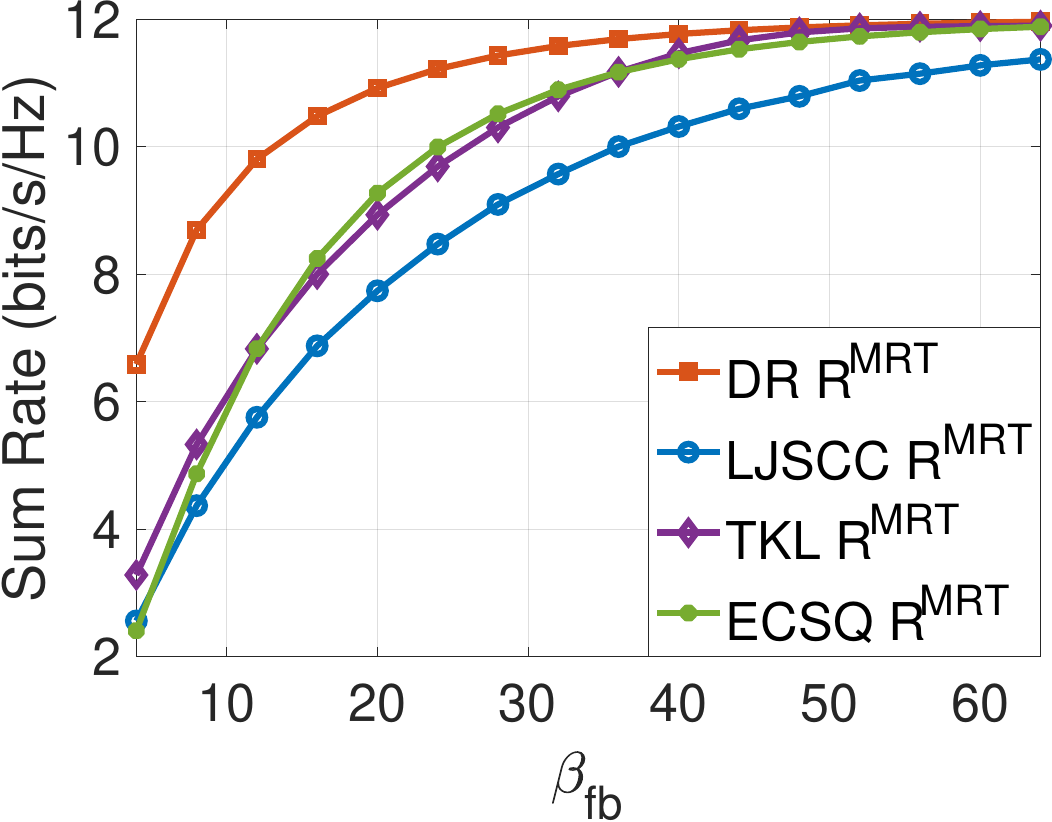}
        \caption{$L = 60$ and ${\rm SNR}= 10~$dB.}
        \label{fig:L_60_snr_10_MRT}
        \end{subfigure}
        ~
        \begin{subfigure}[b]{0.23\textwidth}
        \includegraphics[width=1\columnwidth]{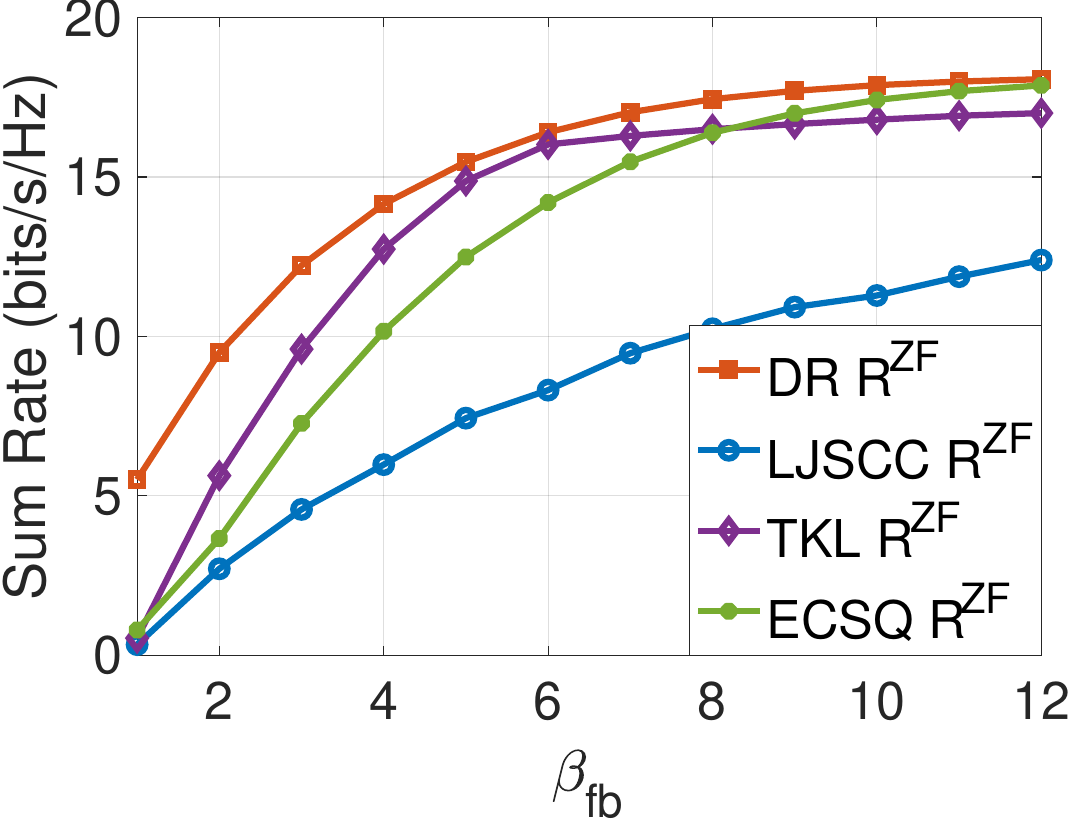}
        \caption{$L = 6$ and ${\rm SNR}= 5~$dB.}
        \label{fig:L_6_snr_5_ZF}
        \end{subfigure}
        ~
        \begin{subfigure}[b]{0.23\textwidth}
        \includegraphics[width=\columnwidth]{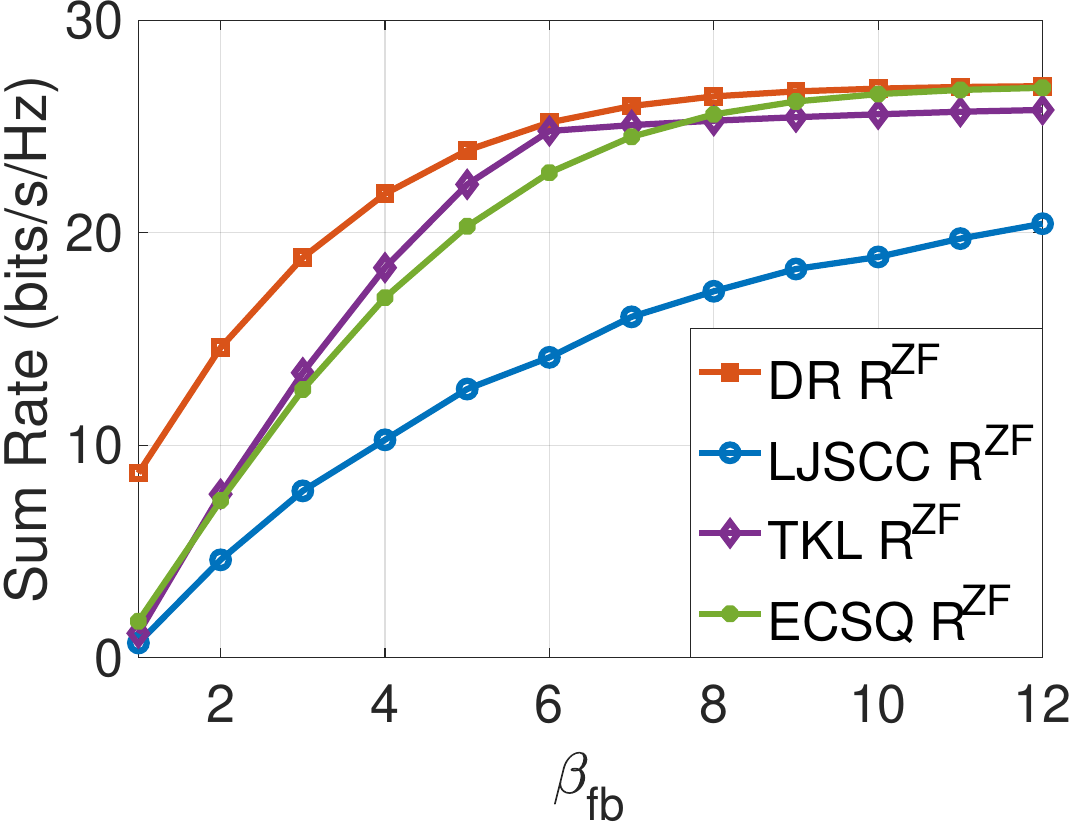}
        \caption{$L = 6$ and ${\rm SNR}= 10~$dB.}
        \label{fig:L_6_snr_10_ZF}
        \end{subfigure}
        ~
        \begin{subfigure}[b]{0.23\textwidth}
        \includegraphics[width=1\columnwidth]{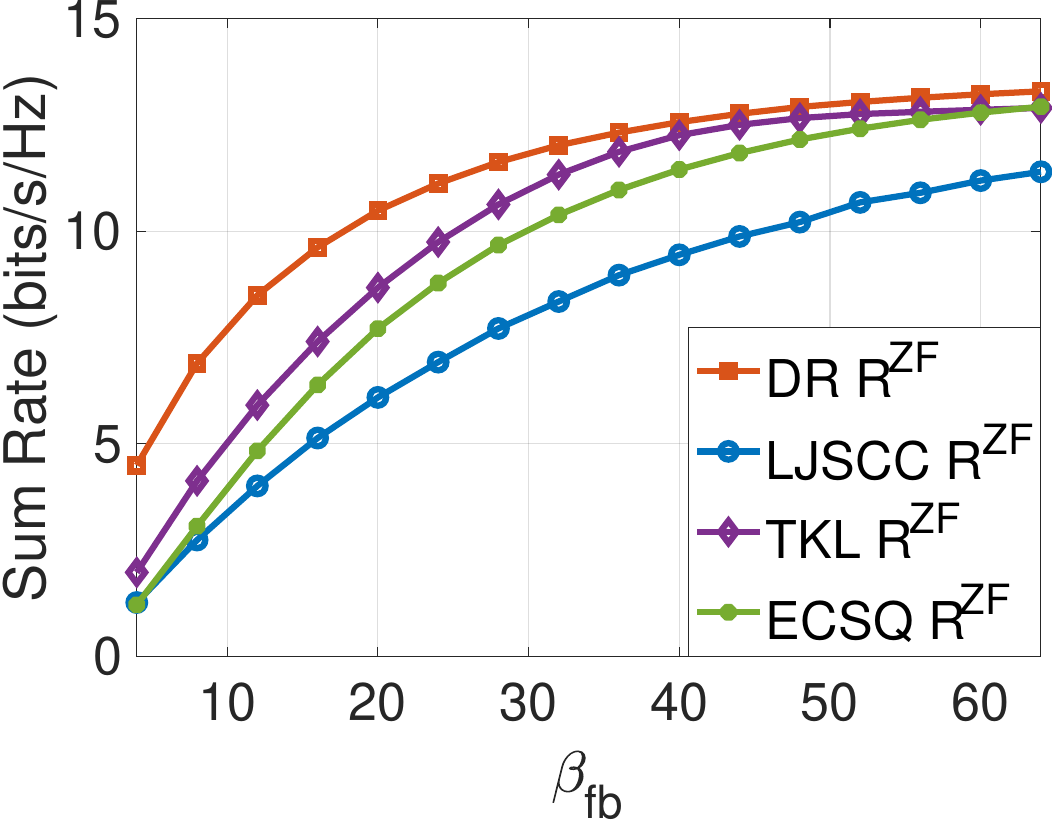}
        \caption{$L = 60$ and ${\rm SNR}= 5~$dB.}
        \label{fig:L_60_snr_5_ZF}
        \end{subfigure}
        ~
        \begin{subfigure}[b]{0.23\textwidth}
        \includegraphics[width=\columnwidth]{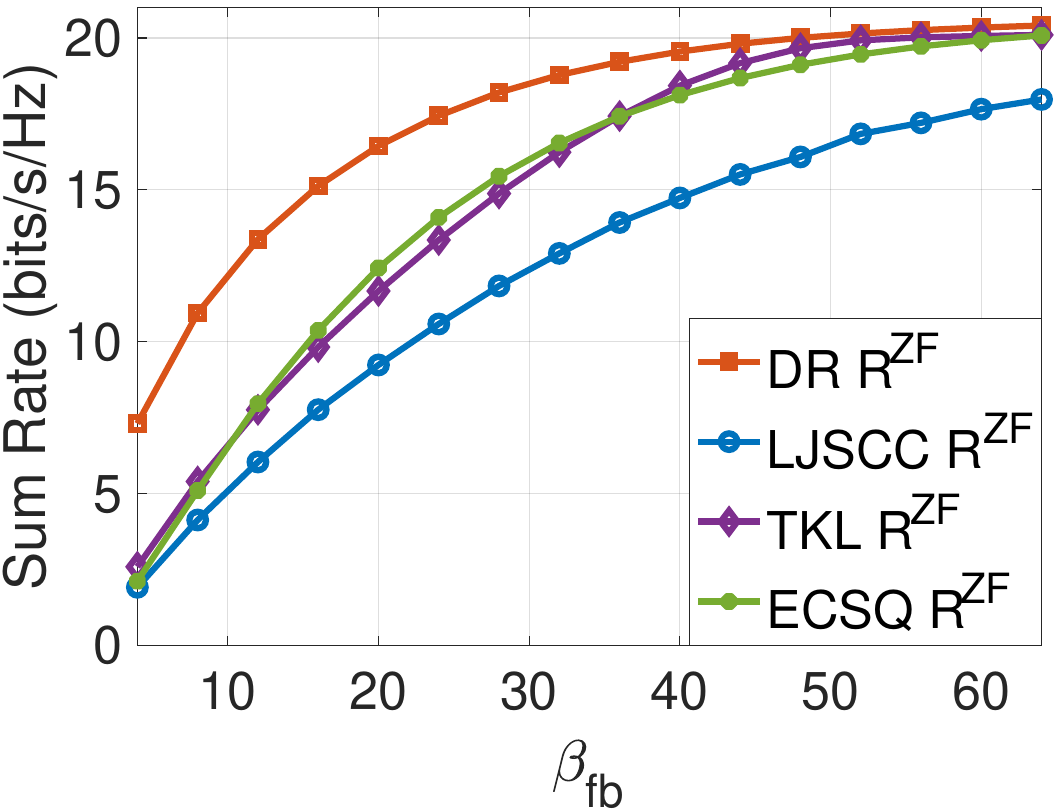}
        \caption{$L = 60$ and ${\rm SNR}= 10~$dB.}
        \label{fig:L_60_snr_10_ZF}
        \end{subfigure}
        \caption{DL sum-rate vs. feedback dimension $\beta_{\rm fb}$ under MRT (a-d) and ZF (e-h) precoding.}
     \label{fig:DL_ZF_beta_fb}
     \vspace{-6mm}
\end{figure*}

\subsection{Effects of Number of Users $K$}
\begin{figure*}[ht!]
\centering
        \begin{subfigure}[b]{0.4\textwidth}
        \includegraphics[width=1\columnwidth]{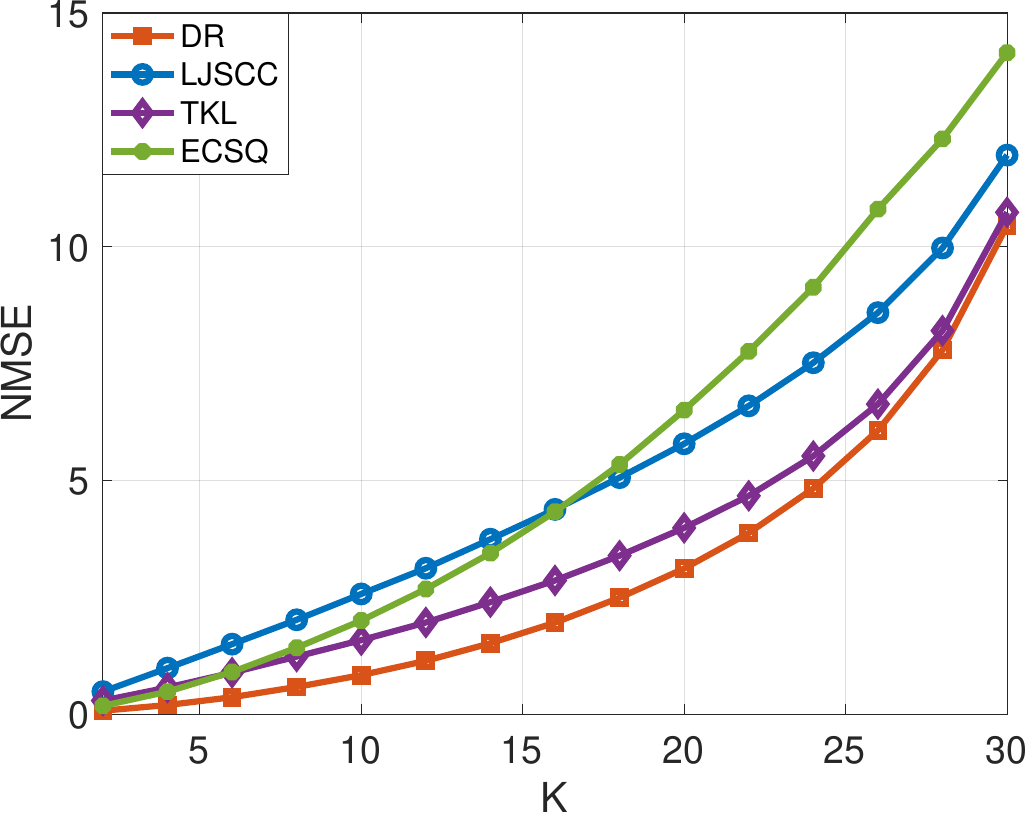}
        \caption{{$L = 6$, $\beta_{\rm fb} = 3$.}}
        \label{fig:L_6_beta_fb_3_user_nmse}
        \end{subfigure}
        ~
        \begin{subfigure}[b]{0.4\textwidth}
        \includegraphics[width=\columnwidth]{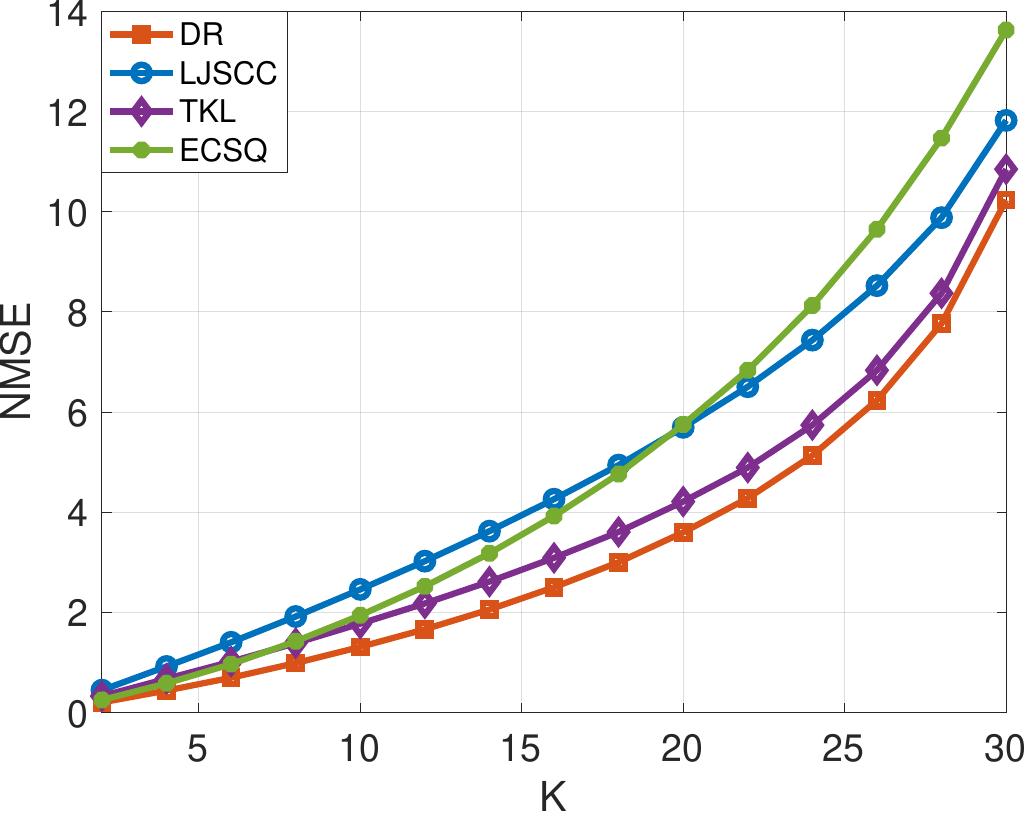}
        \caption{{$L = 60$, $\beta_{\rm fb} = 30$.}}
        \label{fig:L_60_beta_fb_30_user_nmse}
        \end{subfigure}
        ~
        \begin{subfigure}[b]{0.4\textwidth}
        \includegraphics[width=1\columnwidth]{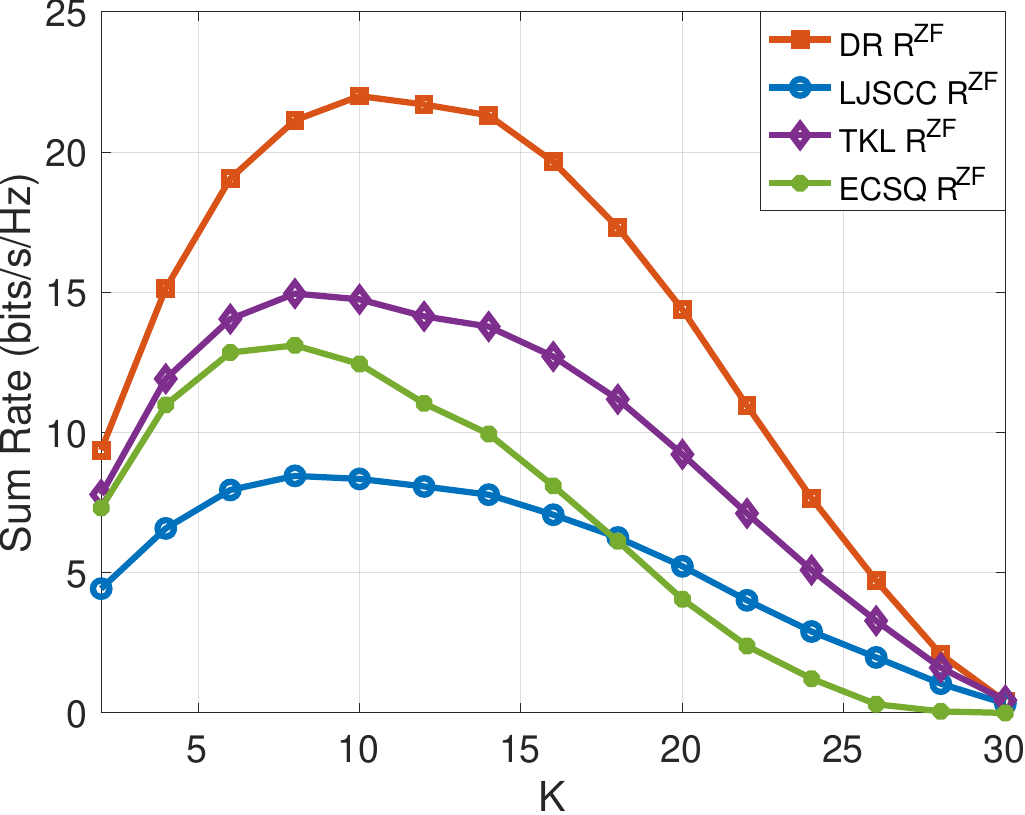}
        \caption{{$L = 6$, $\beta_{\rm fb} = 3$.}}
        \label{fig:L_6_beta_fb_3_user}
        \end{subfigure}
        ~
        \begin{subfigure}[b]{0.4\textwidth}
        \includegraphics[width=\columnwidth]{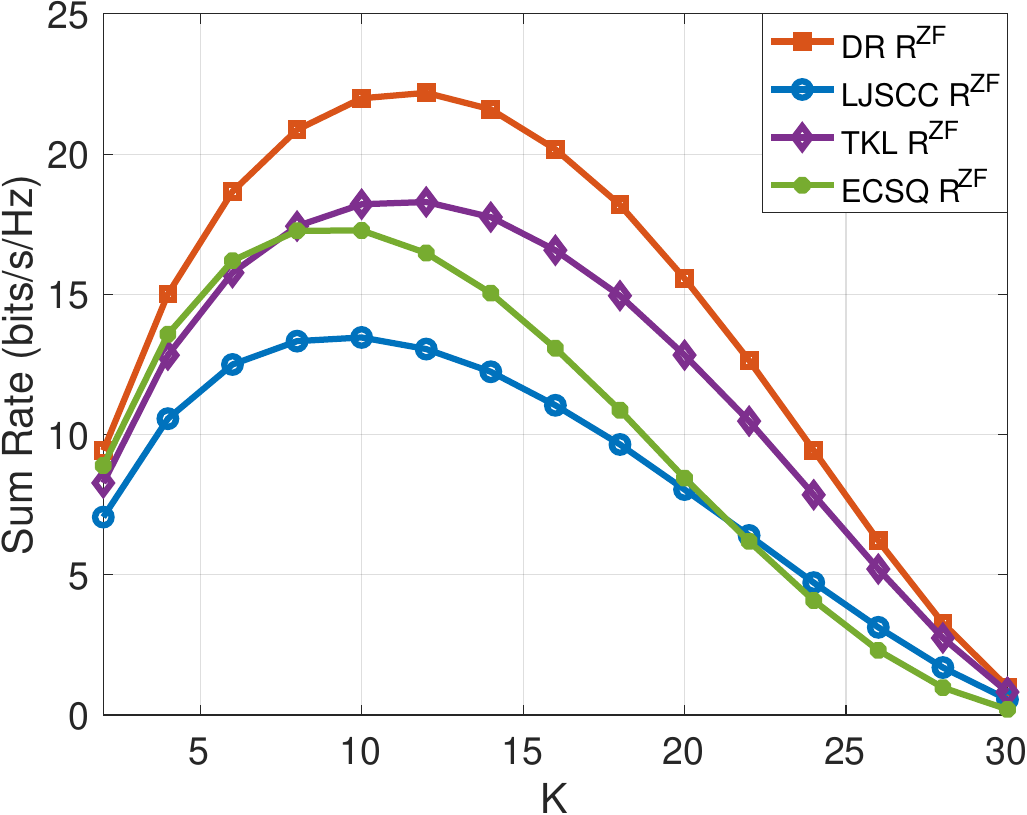}
        \caption{{$L = 60$, $\beta_{\rm fb} = 30$.}}
        \label{fig:L_60_beta_fb_30_user}
        \end{subfigure}
\caption{{Effect of the number of users $K$ under ZF detector and $\text{SNR} = 10$ dB on NMSE (a-b) and DL sum-rate (c-d).}}
     \label{fig:DL_user_number}
     \vspace{-6mm}
\end{figure*}

% \begin{figure*}[ht!]
% \centering
%         \begin{subfigure}[b]{0.4\textwidth}
%         \includegraphics[width=1\columnwidth]{figure3/L_6_user_number.pdf}
%         \caption{\textcolor{blue}{$L = 6$, $\beta_{\rm fb} = 3$.}}
%         \label{fig:L_6_beta_fb_3_user}
%         \end{subfigure}
%         ~
%         \begin{subfigure}[b]{0.4\textwidth}
%         \includegraphics[width=\columnwidth]{figure3/L_60_user_number.pdf}
%         \caption{\textcolor{blue}{$L = 60$, $\beta_{\rm fb} = 30$.}}
%         \label{fig:L_60_beta_fb_30_user}
%         \end{subfigure}
% \caption{{DL sum-rate vs. number of users $K$ under ZF precoding under $\text{SNR} = 10$ dB.}}
%      \label{fig:DL_user_number_rate}
%      % \vspace{-6mm}
% \end{figure*}

{
Next, we investigate the effect of the number of users \(K\) on the NMSE and DL sum rate under the interesting case with a compressed feedback dimension and a practical DL SNR of $10$ dB. We assume the ZF detector at the BS for UL MAC channel, resulting in a multiuser efficiency \(\kappa = 1 - K/M\) \cite{verdu1999spectral}. The results are depicted in Fig.~\ref{fig:DL_user_number} for NMSE in Fig.~\ref{fig:L_6_beta_fb_3_user_nmse}, \ref{fig:L_60_beta_fb_30_user_nmse} and for DL sum rate under ZF precoding in Fig.~\ref{fig:L_6_beta_fb_3_user}, \ref{fig:L_60_beta_fb_30_user}.
From Fig.~\ref{fig:L_6_beta_fb_3_user_nmse} and \ref{fig:L_60_beta_fb_30_user_nmse} we observe that with a larger $K$ the channel NMSE also becomes monotonically larger. This is due to higher interference, i.e., lower multiuser efficiency $\kappa$, resulting in a lower UL channel capacity for each user (see \eqref{eq:C_ul}) and thus the higher channel estimation error. 

Although more users cause worse channel NMSE, it is not necessary to lead a lower DL sum rate, since the summation would have a contribution from more users, see \eqref{eq:sum_rate}. From Fig.~\ref{fig:L_6_beta_fb_3_user} and \ref{fig:L_60_beta_fb_30_user} we see concave curves (rather than the monotonic curves in Fig.~\ref{fig:L_6_beta_fb_3_user_nmse} and \ref{fig:L_60_beta_fb_30_user_nmse} for channel NMSE) of DL sum rate over the number of users $K$, 
% clearly the trade-off of channel estimation error and contribution of more DL rates, 
where with a small number of users, e.g., $K$ is from $1$ to $10$, the DL sum rate rises with increasing $K$, while under $K>10$ the sum rate dramatically drops and eventually converges to almost zero at $K=30$. More specifically, when \(K\) is small, inter-user interference is limited and \(\kappa\) is close to one. As a result, having more users only slightly reduces the channel estimation accuracy for each user, resulting in still satisfied channel estimates, and a growing DL sum rate. However, serving a large number of users leads to dominant inter-user interference in the UL MAC channel, resulting in severe estimation errors and very low individual communication rates so that the sum rate decreases. In the extreme case with \(K = M\) (i.e., the number of users equals the number of antennas), we have \(\kappa = 0\) under ZF detector, i.e., no UL feedback is available, and the sum rate drops to zero.
Additionally, as can be seen in all figures in Fig.~\ref{fig:DL_user_number}, the TKL scheme consistently outperforms the ECSQ and LJSCC schemes, approaching the performance of the DR scheme over a wide range of number of users \(K\). 

{\subsection{Comparison with the Deep Learning-based Method}
Finally, we compare the four presented feedback schemes (the TKL, DR, LJSCC and ECSQ schemes) with the state-of-the-art (SOTA) deep learning-based method, TransNet \cite{cui2022transnet}, under the assumption that each user has direct access to DL CSI and an infinite SNR or infinite bits for UL feedback. For a fair comparison, we generate synthetic training datasets as described in Section \ref{sec:simulation}, consisting of 100,000 channel geometries, each with one channel realization. The TransNet model is retrained using the same hyperparameters as in \cite{cui2022transnet}, including batch size, number of epochs, and learning rate. The evaluation is performed with the following settings: (1) \(L = 6, \beta_{\rm fb} = 3\); (2) \(L = 6, \beta_{\rm fb} = 10\); (3) \(L = 60, \beta_{\rm fb} = 30\); and (4) \(L = 60, \beta_{\rm fb} = 64\). 
With infinite SNR or unlimited feedback bits in the UL channel, the JSCC scheme can transmit the complex-valued codeword of dimension \(\beta_{\rm fb}\) back to the BS without noise, while the SSCC scheme can transmit the user's observations without quantization error.
Consequently, the channel estimation error for the DR and ECSQ schemes is zero across all simulated scenarios. Moreover, when the number of multipaths does not exceed the feedback dimension (\(L \leq \beta_{\rm fb}\)), the TKL and LJSCC schemes also achieve zero channel estimation error. However, to prevent numerical instability in simulations, we assume a high but finite UL SNR of 100 dB for all four schemes.
As shown in Table \ref{tab:nmse_comparison}, TransNet shows much worse performance than any of the covariance-based methods in our paper.  
The key limitation of TransNet is its inability to adapt to specific channel geometries, as it is trained on an ensemble of randomly generated geometries. Moreover, retraining TransNet for each new geometry would be entirely impractical. In contrast, methods that incorporate covariance knowledge achieve performance close to that of the DR scheme.
In addition, learning the channel covariance matrix only requires a small number of channel snapshots \cite{yang2023structured}, highlighting the {dramatic superiority} of our proposed TKL scheme with respect to DNN-based schemes. 
}

% \begin{table}[h]
%     \centering
%         % \vspace{-2mm}
%     % \renewcommand{\arraystretch}{1} % 调整行高
%     % \setlength{\tabcolsep}{0.2pt} % 调整列间距
%     \caption{Comparison of NMSE (dB) with the TransNet \cite{cui2022transnet}.}
%     \label{tab:nmse_comparison}
%     \begin{tabularx}{1\columnwidth} { 
%       >{\centering\arraybackslash}X| 
%       >{\raggedleft\arraybackslash}X
%       >{\raggedleft\arraybackslash}X
%       >{\raggedleft\arraybackslash}X 
%       >{\raggedleft\arraybackslash}X
%       >{\raggedleft\arraybackslash}X }
%     \toprule
%     Setting  & TransNet & TKL & DR & LJSCC & ECSQ \\
%     \midrule
%     $L=6, \betafb=3$  & $-1.1$ & $\mathbf{-5.6}$ & $-58.1$ & $-3.9$ & $-53.6$\\
%  $L=6, \betafb=3$& $-2.4$ & $\mathbf{-116.8}$ & $-156.4$ & $-112.1$ & $-144.3$ \\
%  $L=6, \betafb=3$& $-2.1$   & $\mathbf{-7.0}$& $-59.2$ & $-4.6$& $-54.7$ \\
%  $L=6, \betafb=3$ & $-4.2$  & $\mathbf{-115.1}$ & $-159.4$ & $-102.2$ & $-119.4$ \\
%     \bottomrule
%     \end{tabularx}
%     \vspace{-5mm}
% \end{table}

\begin{table}[h]
    \centering
        % \vspace{-2mm}
    % \renewcommand{\arraystretch}{1} % 调整行高
    % \setlength{\tabcolsep}{0.2pt} % 调整列间距
    \begin{tabular} {l | r r r r r }
    \Xhline{2\arrayrulewidth}
    Setting  & TransNet & TKL & DR & LJSCC & ECSQ \\
    \hline
    $L=6, \betafb=3$  & $-1.1$ & $\mathbf{-5.6}$ & $-58.1$ & $-3.9$ & $-53.6$\\
    $L=6, \betafb=10$  & $-2.4$ & $\mathbf{-116.8}$ & $-156.4$ & $-112.1$ & $-144.3$ \\
    $L=60, \betafb=30$  & $-2.1$   & $\mathbf{-7.0}$& $-59.2$ & $-4.6$& $-54.7$ \\
    $L=60, \betafb=64$  & $-4.2$  & $\mathbf{-115.1}$ & $-159.4$ & $-102.2$ & $-119.4$ \\
   \Xhline{2\arrayrulewidth}
    \end{tabular}
    \caption{{Comparison of NMSE (dB) with the TransNet \cite{cui2022transnet}.}}
    \label{tab:nmse_comparison}
    \vspace{-5mm}
\end{table}

}

\section{Conclusion}
In this paper, we have studied the DL channel estimation problem in the FDD multiuser massive MIMO system. We first provide a lower bound on the CSIT estimation error of any studied JSCC or SSCC feedback schemes based on the remote distortion rate theory. Focusing on the practically demanded case with compressed feedback and minimum delay, we proposed a JSCC feedback scheme (TKL feedback) that exploits the second-order statistics of the channel and analyzed its channel estimation MSE decay behavior in high SNR. For easier evaluation of DL spectral efficiency based on the different feedback schemes, we derived closed-form rate expression under MRT precoding. Our extensive numerical results show that the proposed JSCC scheme outperforms a simple JSCC scheme, and a standard SSCC scheme in terms of both the MSE of the channel estimation and the resulting spectral efficiency under the range of practical SNR. In addition, it outperforms a deep learning-based scheme in an idealized scenario. 
We notice that, in general, the JSCC paradigm breaks the paradigm of layer separation and thus has not found wide applications beyond the theoretical research for general source transmission problems (e.g., video, audio, etc). We believe that CSI feedback is an ideal application for JSCC and advocate that future standards will adopt such type of feedback schemes.
\appendices

 $\vspace{-7mm}$
\section{The solution of \eqref{eq:opt_ppp}}
\label{sec:lemma_TKL}
To solve \eqref{eq:opt_ppp}, we first rewrite the objective more simply. Noticing that the objective can be written as $\trace(\Cm_{\hv_k}) - \trace(\Cm_{\widehat{\hv}_k^{\rm tkl}})$ and the first term is constant to the variable $\alphav_k$, we thus focus on the second term $-\trace(\Cm_{\widehat{\hv}_k^{\rm tkl}}) = -\trace(\Cm_{\hv_k \widehat{\yv}^{\rm fb}_k} \Cm_{\widehat{\yv}^{\rm fb}_k}^{-1} \Cm_{\hv_k \widehat{\yv}^{\rm fb}_k}^{\herm})$. Further noticing that $\Cm_{\widehat{\yv}^{\rm fb}_k} = \diag({\alphav_k})\Sm_k(\widebar{\Lambdam}_k + \Id_{\beta_{\rm tr}})\Sm_k^\herm  +\Id_{\beta_{\rm fb}}
    = \diag\left({\alpha_{k, 1}(\widebar{\lambda}_{k, 1} + 1) + 1}, ..., {\alpha_{k, \beta_{\rm fb}}(\widebar{\lambda}_{k, \beta_{\rm fb}} + 1) + 1}\right)$
and defining $\boldsymbol{\nu}_k = [\nu_{k,1},\dots,\nu_{k,\beta_{\rm fb}}]$ with $\nu_{k,j} \triangleq \frac{\alpha_{k, j} }{\alpha_{k, j}(\widebar{\lambda}_{k, j} + 1) + 1}$ and  $\Deltam_k \triangleq \Sm_k^\herm \diag\left(\boldsymbol{\nu}\right)\Sm_k =
    \diag(\nu_{k,1}, \dots, \nu_{k, \beta_{\rm fb}}, \underbrace{0,\dots, 0}_{\betatr - \beta_{\rm fb}})$,
we have 
\begin{align}
    &\trace(\Cm_{\hv_k \widehat{\yv}^{\rm fb}_k} \Cm_{\widehat{\yv}^{\rm fb}_k}^{-1} \Cm_{\hv_k \widehat{\yv}^{\rm fb}_k}^{\herm})\\ 
        =\;&\trace(\Cm_{\hv_k}\Xm^\herm \Um_k \Deltam_k \Um_k^\herm \Xm \Cm_{\hv_k}) \\
        =\;&\trace(\Cm_{\hv_k}^{\frac{1}{2}}\Cm_{\hv_k}^{\frac{1}{2}}\Xm^\herm \Um_k \Deltam_k \Um_k^\herm \Xm \Cm_{\hv_k}^{\frac{1}{2}}\Cm_{\hv_k}^{\frac{1}{2}}) \\
        =\;& \trace\left(\Cm_{\hv_k}^{\frac{1}{2}} \Vm_k \Sigmam_k^\herm \Deltam_k  \Sigmam_k  \Vm_k^\herm \Cm_{\hv_k}^{\frac{1}{2}}\right) \\
        =\;& \trace\left(  \Vm_k^\herm \Cm_{\hv_k} \Vm_k \Sigmam_k^\herm \Deltam_k  \Sigmam_k \right) \\
        \stackrel{(a)}{=}& \sum_{i=1}^{\beta_{\rm fb}} \frac{\rho_{k,i} \alpha_{k, i} \widebar{\lambda}_{k,i}}{\alpha_{k, i}(\widebar{\lambda}_{k, i} + 1) + 1} \triangleq f(\alphav_k),
\end{align} 
where in (a) we denote $\rho_{k,i} \triangleq \vv_{k, i}^\herm \Cm_{\hv_k} \vv_{k, i}$ with $\vv_{k, i}$ as the $i$-th column of $\Vm_k$ and noting that $[\Vm_k^\herm \Cm_{\hv_k} \Vm_k]_{i,i} = \vv_{k, i}^\herm \Cm_{\hv_k} \vv_{k, i}$. Then, the optimization problem can be equivalently rewritten as
\begin{subequations}\label{eq:opt_new}
\begin{align}
    \underset{\alphav_k\geq0}{\text{minimize}} \;\;\;\;\quad \quad &  -f(\alphav_k) \\
    \text{subject to} \quad \qquad&  \eqref{eq:power_constraint}.
\end{align}
\end{subequations}
The constraints are obviously convex. We show the convexity of the objective by showing that its Hessian matrix is positive semi-definite. It is clear that $\frac{\partial^2 -f(\alphav_k)}{\partial \alpha_{k,i}\alpha_{k,j}} = 0, \forall i \neq j$, and  
\begin{equation}
    \frac{\partial^2 -f(\alphav_k)}{\partial \alpha_{k,i}^2} = \frac{2\rho_{k,i}\bar{\lambda}_{k,i}(\bar{\lambda}_{k,i}+1)(\alpha_{k,i}(\bar{\lambda}_{k,i}+1)+1)}{(\alpha_{k,i}(\bar{\lambda}_{k,i} + 1) + 1)^4}\geq 0.
\end{equation}
Therefore, the Hessian matrix is a diagonal matrix with all non-negative diagonal elements and thus a positive semi-definite matrix. In total, the problem in \eqref{eq:opt_new} is a convex problem. Then, we solve the problem using the Lagrangian method.
Given Lagrangian multiplier $\gamma$, the corresponding Lagrangian function is given as 
\begin{align}\label{eq:lag_func}
    L = -\sum_{i=1}^{\betafb}\frac{\rho_{k,i} \alpha_{k, i}\widebar{\lambda}_{k, i} }{\alpha_{k,i}(\widebar{\lambda}_{k, i} + 1) + 1} + \gamma\left( \sum_{i=1}^{\beta_{\rm fb}} \alpha_{k, i}(\widebar{\lambda}_{k, i} + 1) - P_{\rm ul}\right).
\end{align}

Letting the derivative of $L$ regarding $\alpha_{k,i}$ be zero  
\begin{align}
    \frac{\partial L}{\partial \alpha_{k,i}} = \frac{ - \rho_{k,i}\widebar{\lambda}_{k, i}}{(\alpha_{k,i}(\widebar{\lambda}_{k, i}+ 1)+1)^2} + \gamma(\widebar{\lambda}_{k, i} + 1) \overset{!}{=} 0,
\end{align}
we have 
\begin{align}\label{eq:alpha_star}
    \alpha^\star_{k, i} = \frac{\left[\sqrt{ \frac{\rho_{k,i}\widebar{\lambda}_{k, i}}{\gamma^\star(\widebar{\lambda}_{k, i} + 1)}} -1 \right]_+ }{1 + \widebar{\lambda}_{k, i}}, 
\end{align}
where $\gamma^\star$ is chosen such that 
\begin{align}\label{eq:gamma_results}
    \sum_{i=1}^{\beta_{\rm fb}} \left[ \sqrt{ \frac{\rho_{k,i} \widebar{\lambda}_{k, i}}{\gamma^\star(\widebar{\lambda}_{k, i} + 1)}} -1\right]_+ =P_{\rm ul}.
\end{align}
We obtain the optimal $\gamma^\star$ by solving \eqref{eq:gamma_results} using a binary search. Finally, given the optimal $\gamma^\star$, the optimal power allocation $\alphav_k^\star$ is obtained by \eqref{eq:alpha_star}.

\section{The proof of {{Theorem}} \ref{theorem:QSE}}
\label{sec:QSE_proof}
We first denote the eigen-decomposition of the channel covariance as $\Cm_{\hv_k} = \Um_{\hv_k} \Lambdam_{\hv_k} \Um_{\hv_k}^\herm $, where $\Um_{\hv_k} \in \bC^{MN\times r_k}$ is the tall matrix with eigenvectors and $\Lambdam_{\hv_k}\in \bR^{r_k\times r_k}$ is a diagonal matrix with $r_k$ non-zero eigenvalues $\lambda^{\hv_k}_1 \geq \lambda^{\hv_k}_2\geq\dots \geq \lambda^{\hv_k}_{r_k} $ in descending order on its diagonal. By further denoting $\Psim\triangleq\Um_k\Sm_k^\herm {\rm diag}(\sqrt{\alphav_k})$, we find that the resulting terms of $\Cm_{\hv_k \widehat{\yv}^{\rm fb}_k}$ and $ \Cm_{\widehat{\yv}^{\rm fb}_k}$ have the same formulations as the corresponding terms of the ``analog feedback'' scheme in \cite{khalilsarai2023fdd}. Therefore, following the similar steps in \cite[Appendix B]{khalilsarai2023fdd}, the MSE of channel estimation can be derived as 
\begin{align}\label{eq:MSE_TKL}
    D_k^{\rm tkl}(R) &= \mathbb{E}[\|\hv_k - \widehat{\hv}_k^{\rm tkl}\|^2] \nonumber \\
    &=  \tr\left(\Lambdam_{\hv_k} \left(\Id - \Gm_k + \Gm_k(\Id + \Gm_k)^{-1} \Gm_k\right)\right),
\end{align}
where $\Gm_k$ is defined as 
\begin{align}
    \Gm_k &\triangleq \Lambdam_{\hv_k}^{\frac{1}{2}} \Um_{\hv_k}^\herm \Xm^\herm \Psim (\Psim^\herm \Psim + \Id)^{-1} \Psim^\herm \Xm \Um_{\hv_k}\Lambdam_{\hv_k}^{\frac{1}{2}}.
\end{align}   
Using the same trace inequality on \eqref{eq:MSE_TKL} as in \cite[(29) in Appendix A]{khalilsarai2023fdd}, one can show that 
\begin{align}\label{eq:trace_inequality}
   \lambda^{\hv_k}_{r_k} g({\mathsf{snr}}_{\rm dl}) \leq  D_k^{\rm tkl}(R) \leq \lambda^{\hv_k}_1 g({\mathsf{snr}}_{\rm dl}),
\end{align}
where $g({\mathsf{snr}}_{\rm dl}) \triangleq \tr\left( \Id - \Gm_k + \Gm_k(\Id+ \Gm_k)^{-1} \Gm_k\right)$. 
Then, we show the trend of $g(\cdot)$ in large $\snrdl$.
For a given pilot matrix $\Xm$, denote the eigenvalues of $ \Gm_k$ by $u_k^i, i=[r_k]$, $ g({\mathsf{snr}}_{\rm dl})$ is given by
\begin{align}
    g({\mathsf{snr}}_{\rm dl}) &= r_k - \sum_{i=1}^{r_k} u^i_k + \sum_{i=1}^{r_k} \frac{{u^i_k}^2}{u^i_k + 1} \\
    &=\sum_{i=1}^{r_k} \frac{1}{u^i_k + 1}. \label{eq:gu}
\end{align}
We now derive an approximation of $\Gm_k$ in large $\snrdl$ for further analysis, given as
\begin{align}
    \Gm_k &=  \Lambdam_{\hv_k}^{\frac{1}{2}} \Um_{\hv_k}^\herm \Xm^\herm \Um_k\Sm_k^\herm {\rm diag}(\sqrt{\alphav_k}) ({\rm diag}({\alphav_k})+ \Id)^{-1}  \nonumber \\
    & \quad \quad \quad \quad  {\rm diag}(\sqrt{\alphav_k}) \Sm_k \Um_k^\herm \Xm \Um_{\hv_k}\Lambdam_{\hv_k}^{\frac{1}{2}}   \\
    & \approx \Lambdam_{\hv_k}^{\frac{1}{2}} \Um_{\hv_k}^\herm \Xm^\herm \Um_k\Sm_k^\herm \Sm_k \Um_k^\herm  \Xm  \Um_{\hv_k}\Lambdam_{\hv_k}^{\frac{1}{2}}, \label{eq:G_value}
\end{align}
where the approximation in \eqref{eq:G_value} holds when ${\mathsf{snr}}_{\rm dl} \rightarrow \infty$, which results in $\alphav_k \rightarrow \infty$ based on \eqref{eq:gamma_results}, and $({\rm diag}({\alphav_k})+ \Id)^{-1} \approx {\rm diag}({\alphav_k})^{-1}$. 
Checking the items in \eqref{eq:G_value}, we find out: (a) $\Lambdam_{\hv_k}$ and $ \Um_{\hv_k}$ are independent of ${\mathsf{snr}}_{\rm dl}$; (b) the training matrix can be written as $\Xm = \sqrt{{\mathsf{snr}}_{\rm dl}} \Xm^o$, where $\Xm^o$ is randomly generated and independently from ${\mathsf{snr}}_{\rm dl}$; and (c) $\Sm_k^\herm\Sm_k$ is a diagonal matrix with first $\betafb$ diagonal elements as ones and all other elements as zeros. Therefore, the behavior of the non-zero eigenvalues of $\Gm_k$ in high $\snrdl$ is  
\begin{equation}\label{eq:uk}
    u_k^i = \Theta({\mathsf{snr}}_{\rm dl}), \quad \forall u_k^i \neq 0. 
\end{equation}
Using \eqref{eq:gu} and \eqref{eq:uk}, we deduce that if $\Gm_k$ is full rank, i.e., $u_k^i \neq 0, \forall i \in [r_k]$, we have $g(\snrdl)=\Theta(\snrdl^{-1})$ and then $D_k^{\rm tkl}(R) = \Theta(\snrdl^{-1})$ using \eqref{eq:trace_inequality}. Conversely, if $\Gm_k$ has at least one zero eigenvalue, i.e., $\exists u_k^i = 0$, we have $g(\snrdl) = \Theta(1)$ from \eqref{eq:gu} and then $D_k^{\rm tkl}(R) = \Theta(1)$ using \eqref{eq:trace_inequality}. In short:
\begin{equation}
    D_k^{\rm tkl}(R) = \begin{cases}
        \Theta(\snrdl^{-1}), & u_k^i \neq 0, \forall i \in [r_k],\\
        \Theta(1), & \text{otherwise}.
    \end{cases}
\end{equation}

Finally, we check the rank of $\Gm_k$. It is clear from \eqref{eq:G_value} that when $\betafb \geq r_k$, $\Gm_k$ is full rank with probability one, and thus $D_k^{\rm tkl}(R) = \Theta(\snrdl^{-1})$. Conversely, if $\betafb < r_k$, $\Gm_k$ has a rank at most equal to $\betafb$ for any realization of the training matrix $\Xm$, leading to  $u_k^i=0$ for some $i$ and thus $D_k^{\rm tkl}(R) = \Theta(1)$. This completes the proof.

\section{The proof of Lemma \ref{lemma: MRT}}
\label{sec:MRT_lemma}
Given $\vv_k = \sqrt{\eta} \widehat{\hv}_k, ~\forall k \in [K]$ in MRT precoding, in order to satisfy the power constraint in \eqref{eq: V_power}, it is evident that the power scaling factor $\eta$ can be obtained from \eqref{eq:eta_MRT}.
Next, we derive the three terms $\mathbb{E}[\hv_k^\herm \vv_k]$, ${\rm Var}(\hv_k^\herm \vv_k)$ and $\mathbb{E}[|\hv_k^\herm \vv_j|^2]$ for $j \neq k$ in \eqref{eq:R_hard} as follows. First, we have
\begin{align}
        \mathbb{E}[\hv_k^\herm \vv_k] &= \sqrt{\eta} \mathbb{E}[{\hv_k^\herm \widehat{\hv}_k}] \nonumber \\
        &= \sqrt{\eta} \mathbb{E}[{(\widehat{\hv}_k + \ev_k)^\herm\widehat{\hv}_k}] \label{eq:h_v_k_0} \\
        &= \sqrt{\eta}\mathbb{E}[\widehat{\hv}_k^\herm \widehat{\hv}_k]\nonumber \\
        &=\sqrt{\eta} \trace(\Cm_{\widehat{\hv}_k}), \label{eq:h_v}
\end{align}
where \eqref{eq:h_v_k_0} holds due to \eqref{eq:h_hat}. Then, we have
\begin{align}
        &{\rm Var}(\hv_k^\herm \vv_k)\nonumber \\
        &= \mathbb{E}[(\hv_k^\herm \vv_k -\mathbb{E}[\hv_k^\herm \vv_k])(\hv_k^\herm \vv_k -\mathbb{E}[\hv_k^\herm \vv_k])^\herm] \nonumber \\
        &= \mathbb{E}[{\hv_k^\herm \vv_k\vv_k^\herm \hv_k}] - \eta \trace^2(\Cm_{\widehat{\hv}_k} ) \nonumber \\
        &=\eta \left(\mathbb{E}[{(\widehat{\hv}_k + \ev_k)^\herm \widehat{\hv}_k \widehat{\hv}_k^\herm (\widehat{\hv}_k + \ev_k)}] - \trace^2(\Cm_{\widehat{\hv}_k} ) \right)\nonumber \\
        &= \eta \left(\mathbb{E}[{\widehat{\hv}_k^\herm \widehat{\hv}_k \widehat{\hv}_k^\herm \widehat{\hv}_k}] + \mathbb{E}[{\ev_k^\herm \widehat{\hv}_k \widehat{\hv}_k^\herm \ev_k}] -  \trace^2(\Cm_{\widehat{\hv}_k} )\right) \label{eq:var_hv_0} \\
        &= \eta\left( \trace(\Cm_{\widehat{\hv}_k}^2) +\trace^2(\Cm_{\widehat{\hv}_k}) + \trace(\Cm_{\widehat{\hv}_k}  \Cm_{\ev_k})  - \trace^2(\Cm_{\widehat{\hv}_k} )\right) \label{eq:var_hv_1}\\
        &=\eta \left(\trace(\Cm^2_{\widehat{\hv}_k}) + \trace(\Cm_{\widehat{\hv}_k}  \Cm_{\ev_k})\right)\nonumber \\
        &= \eta \trace(\Cm_{\widehat{\hv}_k} \Cm_{{\hv}_k}),\label{eq:var_h_v}
\end{align}
where \eqref{eq:var_hv_0} holds since $\ev_k$ and $\widehat{\hv}_k$ are uncorrelated, and where the term $ \mathbb{E}[{\widehat{\hv}_k^\herm \widehat{\hv}_k \widehat{\hv}_k^\herm \widehat{\hv}_k}]$  from \eqref{eq:var_hv_0} to \eqref{eq:var_hv_1} is derived as follows. Let $\Cm_{\widehat{\hv}_k} = \Um_{\widehat{\hv}_k} \Lambdam_{\widehat{\hv}_k} \Um_{\widehat{\hv}_k}^{\herm}$ be the eigenvalue decomposition and denote $\widetilde{\hv}_k \triangleq \Um_{\widehat{\hv}_k}^{\herm} \widehat{\hv}_k$, we have $\widetilde{\hv}_k \sim \Cc\Nc(\mathbf{0}, \Lambdam_{\widehat{\hv}_k})$.
Further denoting $\widetilde{h}_{k, i}$ as the $i$-th element of $\widetilde{\hv}_{k}$ and $\lambda_{\widehat{\hv}_k, i}$ as the $i$-th eigenvalue of $\Cm_{\widehat{\hv}_k}$, we have 
\begin{align}
        &\mathbb{E}[{\widehat{\hv}_k^\herm \widehat{\hv}_k \widehat{\hv}_k^\herm \widehat{\hv}_k}]\nonumber \\ 
        &=\mathbb{E}[\|\widetilde{\hv}_k\|^2 \|\widetilde{\hv}_k\|^2] \nonumber \\
        &= \mathbb{E}\left[\left(\sum_{i=1}^{MN} |\widetilde{h}_{k, i}|^2\right)^2\right] \label{eq:4_h_1} \\
        &= \sum_{i=1}^{MN} \mathbb{E}[|\widetilde{h}_{k, i}|^4] + \sum_{i=1}^{MN} \sum_{j=1, j\neq i}^{MN} \mathbb{E}[|\widetilde{h}_{k, i}|^2|\widetilde{h}_{k, j}|^2] \nonumber \\
        &= \sum_{i=1}^{MN} 2 \lambda_{\widehat{\hv}_k, i}^2 + \sum_{i=1}^{MN} \sum_{j=1, j\neq i}^{MN} \lambda_{\widehat{\hv}_k, i} \lambda_{\widehat{\hv}_k, j} \label{eq:4_h_0} \\
        &= \sum_{i=1}^{MN} 2 \lambda_{\widehat{\hv}_k, i}^2 + \sum_{i=1}^{MN} \lambda_{\widehat{\hv}_k, i} \left(\left(\sum_{j=1}^{MN}  \lambda_{\widehat{\hv}_k, j}\right) - \lambda_{\widehat{\hv}_k, i}\right) \nonumber \\
        &= \sum_{i=1}^{MN} 2 \lambda_{\widehat{\hv}_k, i}^2 + \sum_{i=1}^{MN} \lambda_{\widehat{\hv}_k, i} \sum_{j=1}^{MN}  \lambda_{\widehat{\hv}_k, j} - \sum_{i=1}^{MN} \lambda_{\widehat{\hv}_k, i}^2 \nonumber \\
        &= \trace(\Cm_{\widehat{\hv}_k}^2) + \trace^2(\Cm_{\widehat{\hv}_k}),
\end{align}
where \eqref{eq:4_h_0} holds based on [Appendix A.2.4]\cite{mimo2016marzetta}, which states that given a random variable $z \sim \Cc\Nc(0, \lambda)$, the fourth-order moment is given as $\mathbb{E}[|z|^4] = 2\lambda^2$.
Finally, we have
\begin{align}
        \mathbb{E}[|\hv_k^\herm \vv_j|^2] &= \eta\mathbb{E}[|(\widehat{\hv}_k + \ev_k)^\herm \widehat{\hv}_j|^2] \nonumber \\
        &=\eta \mathbb{E}[\widehat{\hv}_j^\herm ( \widehat{\hv}_k + \ev_k) (\widehat{\hv}_k + \ev_k)^\herm \widehat{\hv}_j] \nonumber \\
        &= \eta \left(\mathbb{E}[\widehat{\hv}_j^\herm \widehat{\hv}_k \widehat{\hv}_k^\herm \widehat{\hv}_j] + \mathbb{E}[\widehat{\hv}_j^\herm \ev_k \ev_k^\herm \widehat{\hv}_j] \right)\nonumber \\
        &= \eta\trace(\Cm_{\widehat{\hv}_j} \Cm_{\widehat{\hv}_k} + \Cm_{\widehat{\hv}_j}\Cm_{\ev_k} ) \nonumber \\
        &= \eta\trace( \Cm_{\widehat{\hv}_j} \Cm_{\hv_k}), \label{eq: h_v_j}
\end{align}
where \eqref{eq: h_v_j} holds due to the fact that $\Cm_{\hv_k} = \Cm_{\widehat{\hv}_k} + \Cm_{\ev_k}$.
Combining \eqref{eq:h_v}, \eqref{eq:var_h_v} and \eqref{eq: h_v_j}, the closed-form expression of the UatF rate under MRT precoding is given as \eqref{eq:R_hard_MRT}.

\section{the proof of Lemma \ref{lemma: ZF}}\label{sec:proof_ZF}
Given $\Vm  = \sqrt{\widetilde{\eta}}  \widehat{\Hm} (\widehat{\Hm}^{\herm} \widehat{\Hm})^{-1}$ in ZF precoding, according to the power constraint in \eqref{eq: V_power}, it is evident to obtain the result in \eqref{eq:V_eta_ZF}. 
Next, we derive three terms $\mathbb{E}[\hv_k^\herm \vv_k]$, ${\rm Var}(\hv_k^\herm \vv_k)$ and $\mathbb{E}[|\hv_k^{\herm} \vv_j|^2]$ for $j \neq k$ in \eqref{eq:R_hard}.
Based on the properties of ZF precoding, the inner product of $\widehat{\hv}_i$ and $\vv_j$ is given by 
\begin{equation}\label{eq:prop_ZF}
\widehat{\hv}_i^{\herm} \vv_j =
    \begin{cases}
        \sqrt{\widetilde{\eta}} & \text{if } i = j,\\
        0 & \text{if } i \neq j,
    \end{cases} ~\forall~ i, j \in [K].
\end{equation}
First, we have
\begin{align}
        \mathbb{E}[\hv_k^\herm \vv_k] &= \mathbb{E}[ (\widehat{\hv}_k + \ev_k)^{\herm} \vv_k] \nonumber \\
        &= \mathbb{E}[ \widehat{\hv}_k^{\herm} \vv_k] + \mathbb{E}[ \ev_k^{\herm} \vv_k] \nonumber \\
        &= \sqrt{\widetilde{\eta}}, \label{eq:h_v_zf}
\end{align}
where \eqref{eq:h_v_zf} holds according to \eqref{eq:prop_ZF}. Then, we have
\begin{align}
        &{\rm Var}(\hv_k^\herm \vv_k) = \mathbb{E}[(\hv_k^\herm \vv_k -\mathbb{E}[\hv_k^\herm \vv_k])(\hv_k^\herm \vv_k -\mathbb{E}[\hv_k^\herm \vv_k])^\herm] \nonumber \\
        &= \mathbb{E}[{\hv_k^\herm \vv_k\vv_k^\herm \hv_k}] - (\mathbb{E}[\hv_k^\herm \vv_k])^2 \nonumber \\
        &=  \mathbb{E}[{(\widehat{\hv}_k + \ev_k)^\herm \vv_k\vv_k^\herm (\widehat{\hv}_k + \ev_k)}] - \widetilde{\eta} \nonumber \\
        &= \mathbb{E}[{\ev_k^\herm \vv_k\vv_k^\herm \ev_k}]  +  \mathbb{E}[{\widehat{\hv}_k^\herm \vv_k\vv_k^\herm \widehat{\hv}_k}] - \widetilde{\eta} \label{eq:var_h_v_zf_0} \\
        &=  \mathbb{E}[{\ev_k^\herm \vv_k\vv_k^\herm \ev_k}] \nonumber \\
        &= \tr(\Cm_{\vv_k} \Cm_{\ev_k}), \label{eq:var_h_v_zf}
\end{align}
where $\Cm_{\vv_k} \triangleq \bE[\vv_k\vv_k^\herm]$ and  \eqref{eq:var_h_v_zf_0} holds since $\ev_k$ and $\vv_k$ are uncorrelated. Finally, we have
\begin{align}
        \mathbb{E}[|\hv_k^{\herm} \vv_j|^2] &= \mathbb{E}[|( \widehat{\hv}_k + \ev_k)^{\herm} \vv_j|^2] \nonumber \\
        &= \trace\left(\mathbb{E}[( \widehat{\hv}_k + \ev_k)^{\herm} \vv_j \vv_j^{\herm}( \widehat{\hv}_k + \ev_k)]\right)  \nonumber \\
        % &= \trace\left(\mathbb{E}[\widehat{\hv}_k^{\herm} \vv_j \vv_j^{\herm}(\widehat{\hv}_k + \ev_k)] + \mathbb{E}[\ev_k^{\herm} \vv_j \vv_j^{\herm}(\widehat{\hv}_k + \ev_k)]\right)  \nonumber \\
        &= \trace\left(\mathbb{E}[\ev_k^{\herm} \vv_j \vv_j^{\herm}\ev_k]\right) \label{eq:h_k_v_j_zf_0} \\
        &= \trace(\Cm_{\vv_j} \Cm_{\ev_k}), \label{eq:h_k_v_j_zf}
\end{align}
where \eqref{eq:h_k_v_j_zf_0} holds based on \eqref{eq:prop_ZF}.
Combing \eqref{eq:h_v_zf}, \eqref{eq:var_h_v_zf} and \eqref{eq:h_k_v_j_zf}, the UatF rate under ZF precoding is given by
\begin{align}\label{eq:R_zf_org}
    R_k^{{\rm ZF}} = \log\left(1 + \frac{\widetilde{\eta} }{\tr(\Cm_{\vv_k} \Cm_{\ev_k}) + \sum_{j \neq k} \trace(\Cm_{\vv_j} \Cm_{\ev_k}) + 1}\right),
\end{align}
where the denominator in \eqref{eq:R_zf_org} can be further simplified as  
\begin{align}
     &\tr(\Cm_{\vv_k} \Cm_{\ev_k}) + \sum_{j \neq k} \trace(\Cm_{\vv_j} \Cm_{\ev_k}) = \trace\left(\left(\sum_{k=1}^K \Cm_{\vv_k}\right) \Cm_{\ev_k}\right) \nonumber \\
    &=  \trace\left(\left(\sum_{k=1}^K \mathbb{E}[\vv_k \vv_k^{\herm}]\right)\Cm_{\ev_k}\right) \nonumber \\
     &=  \trace\left(\mathbb{E}\left[\Vm \Vm^{\herm}\right]\Cm_{\ev_k}\right) \nonumber \\
     &=\widetilde{\eta} \trace\left(\mathbb{E}\left[\widehat{\Hm} (\widehat{\Hm}^{\herm} \widehat{\Hm})^{-1}(\widehat{\Hm}^{\herm} \widehat{\Hm})^{-1} \widehat{\Hm}^{\herm}\right]\Cm_{\ev_k}\right) \nonumber \\
     &=\widetilde{\eta} \trace\left(\mathbb{E}\left[\widehat{\Hm} (\widehat{\Hm}^{\herm} \widehat{\Hm})^{-2} \widehat{\Hm}^{\herm}\right]\Cm_{\ev_k}\right) \label{eq:denominator}.
\end{align}
Putting \eqref{eq:denominator} back to \eqref{eq:R_zf_org}, we have the UatF rate expressions under ZF precoding as in  \eqref{eq:R_hard_ZF}.

{\small
	\bibliographystyle{IEEEtran}
	\bibliography{references}
}

\end{document}